\documentclass{amsart}
\usepackage{amssymb,mathrsfs,MnSymbol}
\usepackage{color}
\usepackage{hyperref}

\usepackage{hyperref}
\def\bC {\mathbf{C}}

\def\bR {\mathbf{R}}

\def\fH {\mathfrak{H}}
\def\fS {\mathfrak{S}}

\def\cC {\mathcal{C}}

\def\cF {\mathcal{F}}

\def\cH {\mathcal{H}}

\def\cL {\mathcal{L}}
\def\cM {\mathcal{M}}

\def\cP {\mathcal{P}}

\def\cR {\mathcal{R}}

\def\cT {\mathcal{T}}
\def\cU {\mathcal{U}}

\def\a {{\alpha}}

\def\de {{\delta}}

\def\l {{\lambda}}
\def\L {{\Lambda}}
\def\si {{\sigma}}

\def\om {{\omega}}

\def\d {{\partial}}
\def\grad {{\nabla}}
\def\Dlt {{\Delta}}

\def\rstr {{\big |}}

\def\la {\langle}
\def\ra {\rangle}
\def \La {\bigg\langle}
\def \Ra {\bigg\rangle}

\def\bu {{\bullet}}

\newcommand{\Tr}{\operatorname{trace}}

\newcommand{\Ker}{\operatorname{Ker}}
\newcommand{\IM}{\operatorname{Im}}

\newcommand{\Esssup}{\operatorname{ess sup}}

\newcommand{\ad}{\operatorname{\mathbf{ad}}}

\def\hb {{\hbar}}

\newcommand{\ba}{\begin{aligned}}
\newcommand{\ea}{\end{aligned}}

\newcommand{\be}{\begin{equation}}
\newcommand{\ee}{\end{equation}}

\newcommand{\lb}{\label}

\newtheorem{Thm}{Theorem}[section]

\newtheorem{Cor}[Thm]{Corollary}
\newtheorem{Lem}[Thm]{Lemma}

\begin{document}

\title[Pickl's Proof and Quantum Klimontovich Solutions]{Pickl's Proof of the Quantum Mean-Field Limit and Quantum Klimontovich Solutions}

\author[I. Ben Porat]{Immanuel Ben Porat}
\address[I.B.P.]{CMLS, \'Ecole polytechnique, CNRS, IP Paris, 91128 Palaiseau Cedex, France}
\email{immanuel.ben-porath@polytechnique.edu@polytechnique.edu}

\author[F. Golse]{Fran\c cois Golse}
\address[F.G.]{CMLS, \'Ecole polytechnique, CNRS, IP Paris, 91128 Palaiseau Cedex, France}
\email{francois.golse@polytechnique.edu}

\begin{abstract}
This paper discusses the mean-field limit for the quantum dynamics of $N$ identical bosons in $\bR^3$ interacting via a binary potential with Coulomb type singularity. Our approach is based on the theory of quantum Klimontovich solutions
defined in [F. Golse, T. Paul, Commun. Math. Phys. \textbf{369} (2019), 1021--1053]. Our first main result is a definition of the interaction nonlinearity in the equation governing the dynamics of quantum Klimontovich solutions for a class of
interaction potentials slightly less general than those considered in [T. Kato, Trans. Amer. Math. Soc. \textbf{70} (1951), 195--211]. Our second main result is a new operator inequality satisfied by the quantum Klimontovich solution in the case
of an interaction potential with Coulomb type singularity. When evaluated on an initial bosonic pure state, this operator inequality reduces to a Gronwall inequality for a functional introduced in [P. Pickl, Lett. Math. Phys. \textbf{97} (2011), 
151--164], resulting in a convergence rate estimate for the quantum mean-field limit leading to the time-dependent Hartree equation.
\end{abstract}



\date{\today}

\maketitle

%%%%%%%%%%%%%%%%%%%%%%%%%%%%%%%%%%%%%%%%%%%%%%%%%%%%%%%%%%%%%%%%%%%%%%%%%%%%%%%%%%%%%%%%%%%%%%%%%%%%%%%%%%%%%%%%%%%%%%%%%

\section{Introduction and Notation}\lb{S-Intro}

%%%%%%%%%%%%%%%%%%%%%%%%%%%%%%%%%%%%%%%%%%%%%%%%%%%%%%%%%%%%%%%%%%%%%%%%%%%%%%%%%%%%%%%%%%%%%%%%%%%%%%%%%%%%%%%%%%%%%%%%%
In classical mechanics, the motion equations for a system of $N$ identical point particles of mass $m$ with positions $q_j(t)\in\bR^3$ and momenta $p_j(t)\in\bR^3$ for all $j=1,\ldots,N$ is
\be\lb{NBodyClass}
\left\{
\ba
\,\dot q_j(t)=&\tfrac1m p_j(t)=\grad_{p_j}H_N(p_1(t),\ldots,q_N(t))\,,
\\
\,\dot p_j(t)=&-\sum_{k=1\atop k\not=j}^N\grad V(q_j(t)-q_k(t))=-\grad_{p_j}H_N(p_1(t),\ldots,q_N(t))\,,
\ea
\right.
\ee
where the $N$-particle classical Hamiltonian is
$$
H_N(p_1,\ldots,q_N):=\sum_{j=1}^N\tfrac1{2m}|p_j|^2+\sum_{1\le j<k\le N}V(q_j-q_k)\,.
$$
Assuming that $V\in C^{1,1}(\bR^3)$, this differential system has a unique global solution for all initial data. If $V$ is even\footnote{In this case, the exclusion $j\not=k$ in the right-hand side of Newton's second law for $\dot p_j(t)$ is useless
since $V$ even $\implies\grad V(0)=0$.}, the phase space empirical measure
\be\lb{EmpirClass}
\mu_N(t,dxd\xi):=\tfrac1N\sum_{j=1}^N\de_{q_j(t/N),p_j(t/N)}(dxd\xi)\,,\qquad Nm=1\,,
\ee
is an \textit{exact}, weak solution of the Vlasov equation
\be\lb{KlimoClass}
\d_t\mu_N+\xi\cdot\grad_x\mu_N-\grad_x(V\star\mu_N(t))\cdot\grad_\xi\mu_N=0
\ee
with self-consistent, mean-field potential
$$
V\star_{x,\xi}\mu_N(t,x)=\tfrac1N\sum_{k=1}^NV(x-q_k(t))
$$
This remarkable observation is due to Klimontovich, and solutions of the Vlasov equation \eqref{KlimoClass} of the form \eqref{EmpirClass} are referred to as ``Klimontovich solutions''. Thus, if $\mu_N(0)\to f^{in}dxd\xi$ weakly in the sense 
of probability measures as $N\to\infty$, where $f^{in}$ is a probability density on $\bR^3_x\times\bR^3_\xi$, one has
$$
\mu_N(t,dxd\xi)\to f(t,x,\xi)dxd\xi\text{ weakly in the sense of probability measures}
$$
for all $t\ge 0$ as $N\to\infty$, where $f$ is the solution of the Vlasov equation
\be\lb{Vlasov}
\d_tf+\xi\cdot\grad_xf-\grad _x(V\star_{x,\xi}f(t,\cdot,\cdot))\cdot\nabla_\xi f=0\,,\qquad f\rstr_{t=0}=f^{in}\,.
\ee
Thus, the mean-field limit in classical mechanics is equivalent to the continuous dependence for the weak topology of probability measures of solutions of the Vlasov equation in terms of their initial data. See \cite{BraunHepp} for a proof of 
this result. For instance, the weak convergence of the initial data can be realized by a random choice of $(q_j(0),p_j(0))$, independent and identically distributed with distribution $f^{in}$.

\smallskip
The mean-field limit for bosonic systems in quantum mechanics has been formulated in different settings, by using the so-called BBGKY hierarchy \cite{Spohn80,BGM,BEGMY,EY}, or in the second quantization setting \cite{RodSchlein}. 
Interestingly, these techniques allow considering singular potentials such as the Coulomb potential, instead of $C^{1,1}$ potentials as in the classical case. (The mean-field limit with Coulomb potentials in classical mechanics is still an 
open problem at the time of this writing; see however \cite{Serfaty} in the special case of monokinetic particle distributions. See also \cite{HaurayJabin1,HaurayJabin2} for potentials less singular than the Coulomb potential). 

The quantum mean-field equation analogous to the Vlasov equation \eqref{Vlasov} is the (time-dependent) Hartree equation
\be\lb{TDHWf}
i\hb\d_t\psi(t,x)=-\tfrac12\hbar^2\Dlt_x\psi(t,x)+(V\star|\psi(t,\cdot)|^2)(x)\psi(t,x)=0\,,\qquad\psi\rstr_{t=0}=\psi^{in}\,.
\ee

In \cite{Pickl1,KPickl}, an original method, close to the second quantization approach in \cite{RodSchlein}, but avoiding the rather heavy formalism of Fock spaces, was proposed and successfully applied to singular potentials including
the Coulomb potential.

\smallskip
All these approaches noticeably differ from the classical setting used in \cite{BraunHepp} for lack of a quantum notion of phase-space empirical measures. However, a quantum analogue of the notion of phase-space empirical measure
was recently proposed in \cite{FGTPaul}, along with an equation analogous to \eqref{KlimoClass} governing their evolution. This notion was used in \cite{FGTPaul} to prove the uniformity of the mean-field limit in the Planck constant 
$\hbar>0$. However, the discussion in \cite{FGTPaul} only considers regular potentials (specifically $\d^\a V\in\cF L^1(\bR^d)$ for $|\a|\le 3+[d/2]$). Even writing the equation analogous to \eqref{KlimoClass} satisfied by the quantum
analogue of the phase-space empirical measure requires $V\in\cF L^1(\bR^d)$ in the setting of \cite{FGTPaul}.

\smallskip
The purpose of the present paper is twofold:

\noindent
(a) to extend the formalism of quantum empirical measures considered in \cite{FGTPaul} to treat the case of singular potentials including the Coulomb potential, which is of particular interest for applications to atomic physics (see 
Theorem \ref{T-DefC} in section \ref{S-ExtDef}), and

\noindent
(b) to explain how the ideas in \cite{Pickl1,KPickl} can be couched in terms of the formalism of quantum empirical measures defined in \cite{FGTPaul} (see Theorem \ref{T-OpIneq} and Corollary \ref{C-MF} in section \ref{S-Main}).

Specifically, we prove an inequality \textit{between operators} on the $N$-particle Hilbert space, of which the key estimates in \cite{Pickl1,KPickl} leading to the quantum mean-field limit are straightforward consequences.

\smallskip
The next section briefly recalls only the essential part of \cite{FGTPaul} used in the sequel. The main results obtained in the present paper are Theorems \ref{T-DefC} and \ref{T-OpIneq} from sections \ref{S-ExtDef} and \ref{S-Main}
respectively. The proofs of these results are given in the subsequent sections.

%%%%%%%%%%%%%%%%%%%%%%%%%%%%%%%%%%%%%%%%%%%%%%%%%%%%%%%%%%%%%%%%%%%%%%%%%%%%%%%%%%%%%%%%%%%%%%%%%%%%%%%%%%%%%%%%%%%%%%%%%

\section{Quantum Klimontovich Solutions}\lb{S-QKlim}

%%%%%%%%%%%%%%%%%%%%%%%%%%%%%%%%%%%%%%%%%%%%%%%%%%%%%%%%%%%%%%%%%%%%%%%%%%%%%%%%%%%%%%%%%%%%%%%%%%%%%%%%%%%%%%%%%%%%%%%%%

Consider the quantum $N$-body Hamiltonian
\be\lb{N-Hamilt}
\cH_N:=\sum_{j=1}^N-\tfrac12\hb^2\Dlt_{x_j}+\tfrac1N\sum_{1\le j<k\le N}V(x_j-x_k)
\ee
on $\fH_N:=\fH^{\otimes N}\simeq L^2(\bR^{3N})$, where $\fH:=L^2(\bR^3)$. Henceforth it is assumed that $V$ is a real-valued function such that $\cH_N$ has a (unique) self-adjoint extension to $\fH_N$, still denoted $\cH_N$. 
A well-known sufficient condition for this to be true has been found by Kato (see condition (5) in \cite{Kato51}): there exists $R>0$ such that
\be\lb{KatoCond0V}
\int_{|z|\le R}V(z)^2dz+\Esssup_{|z|>R}|V(z)|<\infty\,.
\ee 
In particular, these conditions include the (repulsive) Coulomb potential in $\mathbf R^3$. In fact, $\cH_N$ has a self-adjoint extension to $\fH_N$ under a condition slightly more general than Kato's original assumption 
recalled above:
\be\lb{KatoCondV}
V\in L^2(\bR^3)+L^\infty(\bR^3)
\ee
(see Theorem X.16 and Example 2 in \cite{RS2}, and Theorem V.9 with $m=1$ in \cite{Penz}).

In the sequel, we adopt the notation in \cite{FGTPaul}. In particular, we set
\be\lb{DefJk}
J_kA:=I_\fH^{\otimes (k-1)}\otimes A\otimes I_\fH^{\otimes(N-k)}\,,\quad 1\le n\le N\,,
\ee
and
\be\lb{DefMNin}
\cM_N^{in}:=\tfrac1N\sum_{k=1}^NJ_k\in\cL(\cL(\fH),\cL(\fH_N))\,.
\ee
The dynamics of the morphism $\cM_N^{in}$ is defined by conjugation with the $N$-particle dynamics as follows: for each $A\in\cL(\fH)$, 
\be\lb{DefMNt}
\cM_N(t)A:=\cU_N(t)^*(\cM_N^{in}A)\cU_N(t)\,,\quad\text{ with }\cU_N(t):=\exp(-it\cH_N/\hb)\,.
\ee
Since $\cH_N$ is self-adjoint, $t\mapsto \cU_N(t)$ is a unitary group by Stone's theorem. The time-dependent morphism $t\mapsto\cM_N(t)\in\cL(\cL(\fH),\cL(\fH_N))$ is henceforth referred to as \textit{the quantum Klimontovich solution}.

Assume henceforth that $V$ is even:
\be\lb{Veven}
V(x)=V(-x)\,,\qquad x\in\bR^d\,.
\ee

The first main result in \cite{FGTPaul} (Theorem 3.3) is that, if $\hat V\in L^1(\bR^d)$, the quantum Klimontovich solution $\cM_N(t)$ satisfies
\be\lb{KlimEq}
i\hb\d_t\cM_N(t)=\ad^*(K)\cM_N(t)-\cC(V,\cM_N(t),\cM_N(t))\,,
\ee
where $K=-\tfrac12\hb^2\Dlt$ is the quantum kinetic energy, and where
\be\lb{Def-ad*}
(\ad^*(T)\Lambda)A:=-\Lambda([T,A])
\ee
for each unbounded self-adjoint operator $T$ on $\fH$, each $A\in\cL(\fH)$ satisfying the condition $[T,A]\in\cL(\fH)$, and each $\Lambda\in\cL(\cL(\fH),\cL(\fH_N))$. Moreover
\be\lb{DefC}
\cC(V,\Lambda_1,\Lambda_2)(A):=\tfrac1{(2\pi)^d}\int_{\bR^d}\hat V(\om)((\Lambda_1E_\om^*)\Lambda_2(E_\om A)-\Lambda_2(AE_\om)(\Lambda_1E_\om^*))d\om
\ee
for each $A\in\cL(\fH)$ and each $\Lambda_1,\Lambda_2\in\cL(\cL(\fH),\cL(\fH_N))$, where $E_\om\in\cL(\fH)$ is the operator defined by 
\be\lb{DefE}
(E_\om\phi)(x):=e^{i\om\cdot x}\phi(x)\quad\text{ for each }\phi\in\fH\text{ and }\om\in\bR^d\,.
\ee
Since the integrand of the right-hand side of \eqref{DefC} takes its values in the non separable space $\cL(\fH_N)$, it is worth mentioning that this integral is a weak integral for the ultraweak topology in $\cL(\fH_N)$ (see footnote 3
on p. 1032 in \cite{FGTPaul}).

At variance with the classical case recalled in \eqref{KlimoClass}, the differential equation \eqref{KlimEq} satisfied by the quantum Klimontovich solution $t\mapsto\cM_N(t)$ is not formally identical to the mean-field, time-dependent Hartree 
equation \eqref{TDHWf}. The relation between \eqref{TDHWf} and \eqref{KlimEq} is explained in Theorem 3.5, the second main result in \cite{FGTPaul}, recalled below.

If $\psi$ is a solution of the the time-dependent Hartree equation \eqref{TDHWf} satisfying the normalization condition
$$
\|\psi(t,\cdot)\|_\fH=1\qquad\text{for all }t\in\bR\,,
$$
the time-dependent morphism $t\mapsto\cR(t)\in\cL(\cL(\fH),\cL(\fH_N))$ defined by the formula\footnote{Throughout this paper, we adopt the Dirac bra-ket notation. Thus a wave function $\psi\in\fH$ viewed as a vector of the linear space 
$\fH$ is denoted $|\psi\ra$, whereas $\la\psi|$ designates the linear functional
$$
\la\psi|:\,\,\fH\ni\phi\mapsto\int_{\bR^d}\overline{\psi(x)}\phi(x)dx\in\bC\,.
$$
If $A\in\cL(\fH)$, we denote
$$
\la\psi|A|\phi\ra:=\int_{\bR^d}\overline{\psi(x)}(A\phi)(x)dx
$$
and $\la\psi|\phi\ra:=\la\psi| I_\fH|\phi\ra$ is the inner product on $\fH$.}
$$
\cR(t)A:=\la\psi(t,\cdot)|A|\psi(t,\cdot)\ra I_{\fH_N}
$$
is a solution of \eqref{KlimEq}.

%%%%%%%%%%%%%%%%%%%%%%%%%%%%%%%%%%%%%%%%%%%%%%%%%%%%%%%%%%%%%%%%%%%%%%%%%%%%%%%%%%%%%%%%%%%%%%%%%%%%%%%%%%%%%%%%%%%%%%%%%

\section{Extending the Definition of $\cC(V,\cM_N(t),\cM_N(t))$ when $V\notin\cF L^1(\bR^3)$}\lb{S-ExtDef}

%%%%%%%%%%%%%%%%%%%%%%%%%%%%%%%%%%%%%%%%%%%%%%%%%%%%%%%%%%%%%%%%%%%%%%%%%%%%%%%%%%%%%%%%%%%%%%%%%%%%%%%%%%%%%%%%%%%%%%%%%

Our first task is to extend the definition \eqref{DefC} of the term $\cC(V,\cM_N(t),\cM_N(t))$ to a more general class of potentials $V$, including the Coulomb potential in $\bR^3$.

Since
$$
\ba
\cM_N(t)(E^*_\om)\cM_N(t)(E_\om|\phi\ra\la\phi|)-\cM_N(t)(|\phi\ra\la\phi|E_\om)\cM_N(t)(E^*_\om)
\\
=\cU_N(t)^*(\cM_N^{in}(E^*_\om)\cM_N^{in}(E_\om|\phi\ra\la\phi|)-\cM_N^{in}(|\phi\ra\la\phi|E_\om)\cM_N^{in}(E^*_\om))\cU_N(t)&\,,
\ea
$$
the idea is to define
$$
\ba
\la\Phi_N^{in}|\cC(V,\cM_N(t),\cM_N(t))(|\phi\ra\la\phi|)|\Psi_N^{in}\ra
\\
:=\tfrac1{(2\pi)^3}\int_{\bR^3}\hat V(\om)\la\cU_N(t)\Phi_N^{in}|S_N[\phi](\om)|\cU_N(t)\Psi_N^{in}\ra d\om
\ea
$$
for all $\Phi_N,\Psi_N^{in}\in\fH_N$, where
$$
S_N[\phi](\om):=\cM_N^{in}(E^*_\om)\cM_N^{in}(E_\om|\phi\ra\la\phi|)-\cM_N^{in}(|\phi\ra\la\phi|E_\om)\cM_N^{in}(E^*_\om)
$$
and to take advantage of the decay of $S_N[\phi]$ in $\om$, assuming that $\phi$ is regular enough. Our argument does not use any regularity on $\Phi_N^{in}$ or $\Psi_N^{in}$. This is quite natural, since anyway Kato's condition 
\eqref{KatoCondV} on the interaction potential $V$ does not entail higher than (Sobolev) $H^2$ regularity for $\cU_N(t)\Phi_N^{in}$ or $\cU_N(t)\Psi_N^{in}$, as observed in Note V.10 of \cite{Penz}.

Our first main result in this paper is the following result, leading to a definition of $\cC(V,\cM_N(t),\cM_N(t))(|\phi\ra\la\phi|)$ in the case of singular, Coulomb-like potentials $V$, and for bounded wave functions $\phi$. This  theorem
can be regarded as an extension to the case of singular, Coulomb like potentials $V$ of the formalism of quantum Klimontovich solutions in \cite{FGTPaul}.

\begin{Thm}\lb{T-DefC}
Assume that $V$ is a real-valued measurable function on $\bR^3$ satisfying the parity condition \eqref{Veven}, and
\be\lb{CondV1}
V\in L^2(\bR^3)+\cF L^1(\bR^3)\,.
\ee
For each $\phi\in L^2\cap L^\infty(\bR^3)$ and each $\Psi_N\in\fH_N$, the function
$$
\om\mapsto\la\Psi_N|S_N[\phi](\om)|\Psi_N\ra\text{ belongs to }L^2\cap L^\infty(\bR^3)\,.
$$
The interaction operator $\cC(V,\cM_N(t),\cM_N(t))(|\phi\ra\la\phi|)$ is defined by the formula
$$
\cC(V,\cM_N(t),\cM_N(t))(|\phi\ra\la\phi|):=\tfrac1{(2\pi)^3}\int_{\bR^3}\hat V(\om)\cU_N(t)^*S_N[\phi](\om)\cU_N(t)d\om
$$
The integral on the right-hand side of the equality above is to be understood as a weak integral and defines 
$$
t\mapsto\cC(V,\cM_N(t),\cM_N(t))(|\phi\ra\la\phi|)
$$ 
as a continuous map from $\bR$ to $\cL(\fH_N)$ endowed with the ultraweak topology, which is moreover bounded on $\bR$ for the operator norm on $\cL(\fH_N)$.
\end{Thm}

Obviously, condition \eqref{CondV1} is stronger than Kato's condition \eqref{KatoCondV}. However, the repulsive Coulomb potential $z\mapsto1/|z|$ in $\bR^3$ obviously satisfies \eqref{CondV1}, since its Fourier transform $\zeta\mapsto{C}/|\zeta|^2$
belongs to $L^1(\bR^3)+L^2(\bR^3)$. In particular, $\cH_N$ has a self-adjoint extension to $\fH_N$ under condition \eqref{CondV1}

\begin{proof}
Assuming that $\Psi_N^{in}\in\fH_N$, one has
$$
\cU_N(t)\Psi_N^{in}\in\fH_N\text{ with }\|\cU_N(t)\Psi_N^{in}\|_{\fH_N}=\|\Psi_N^{in}\|_{\fH_N}\,.
$$

Therefore, we henceforth forget the time dependence in $\Psi_N(t,\cdot)=\cU_N(t)\Psi_N^{in}$, which will be henceforth denoted $\Psi_N\equiv\Psi_N(x_1,\ldots,x_N)$.

Observe first that
$$
S_N[\phi](\om)=\tfrac1{N^2}\sum_{1\le k\not=l\le N}(J_k(E^*_\om)J_l(E_\om|\phi\ra\la\phi|)-J_l(|\phi\ra\la\phi|E_\om)J_k(E^*_\om))
$$
since
$$
\ba
J_k(E^*_\om)J_k(E_\om|\phi\ra\la\phi|)-J_k(|\phi\ra\la\phi|E_\om)J_k(E^*_\om)
\\
=J_k(E^*_\om E_\om|\phi\ra\la\phi|)-J_k(|\phi\ra\la\phi|E_\om E^*_\om)
\\
=J_k(|\phi\ra\la\phi|)-J_k(|\phi\ra\la\phi|)=0&\,.
\ea
$$
Without loss of generality, consider the term
$$
\ba
\la\Psi_N|J_1(E_\om^*)J_2(E_\om|\phi\ra\la\phi|)|\Psi_N\ra=\int_{\bR^6}\!e^{-i\om\cdot(x_1\!-x_2)}\phi(x_2)
\\
\times\left(\int_{\bR^{3N-6}}\!\!\!\overline{\Psi_N(x_1,x_2,Z)}\left(\int_{\bR^3}\!\Psi_N(x_1,y_2,Z)\overline{\phi(y_2)}dy_2\!\right)dZ\!\right)dx_1dx_2
\\
=\hat F(-\om)
\ea
$$
where
$$
F(X):=\int_{\bR^3}\!\phi(X\!+\!x_1)f(x_1,x_1+X,x_1)dx_1\,,
$$
with the notation
$$
f(x_1,x_2,y_1):=\int_{\bR^{3N-6}}\!\!\overline{\Psi_N(x_1,x_2,Z)}\left(\int_{\bR^3}\!\Psi_N(y_1,y_2,Z)\overline{\phi(y_2)}dy_2\!\right)dZ\,.
$$
We shall prove that $F\in L^1(\bR^3)\cap L^2(\bR^3)$, so that $\hat F\in L^2(\bR^3)\cap C_0(\bR^3)$.

First
$$
\ba
\int_{\bR^3}|F(X)|dX\le&\int_{\bR^6}|\phi(X+x_1)||f(x_1,x_1+X,x_1)|dx_1dX
\\
=&\int_{\bR^3}\left(\int_{\bR^3}|\phi(x_2)||f(x_1,x_2,x_1)|dx_2\right)dx_2
\\
\le&\int_{\bR^{3N-3}}\left(\int_{\bR^3}|\phi(x_2)||\Psi_N(x_1,x_2,Z)|dx_2\right)
\\
&\qquad\quad\times\left(\int_{\bR^3}|\Psi_N(x_1,y_2,Z)||\phi(y_2)|dy_2\!\right)dZdx_1
\\
=&\int_{\bR^{3N-3}}\left(\int_{\bR^3}|\phi(x_2)||\Psi_N(x_1,x_2,Z)|dx_2\right)^2dZdx_1
\\
\le&\|\phi\|_{L^2(\bR^3)}^2\int_{\bR^{3N-3}}\int_{\bR^3}|\Psi_N(x_1,x_2,Z)|^2dx_2dZdx_1
\\
=&\|\phi\|_{L^2(\bR^3)}^2\|\Psi_N\|^2_{L^2(\bR^{3N})}<\infty\,.
\ea
$$
where the last inequality is the Cauchy-Schwarz inequality for the inner integral.

On the other hand
$$
\int_{\bR^3}|F(X)|^2dX\le\|\phi\|^2_{L^\infty(\bR^3)}\int_{\bR^3}\left(\int_{\bR^3}|f(x_1,x_1+X,x_1)|dx_1\right)^2dX\,,
$$
and
$$
\int_{\bR^3}|f(x_1,x_1+X,x_1)|dx_1\le\int_{\bR^{3N-3}}|\Psi_N(x_1,x_1+X,Z)|\Phi_N(x_1,Z)dZdx_1\,,
$$
with
$$
\Phi(x_1,Z):=\int_{\bR^3}|\Psi_N(x_1,y_2,Z)||\phi(y_2)|dy_2\,,
$$
so that
$$
\Phi_N(x_1,Z_N)^2\le\|\phi\|^2_{L^2(\bR^3)}\int_{\bR^3}|\Psi_N(x_1,y_2,Z)|^2dy_2\,,
$$
and
$$
\ba
\int_{\bR^{3N-3}}\Phi_N(x_1,Z)^2dZdx_1\le&\|\phi\|^2_{L^2(\bR^3)}\int_{\bR^3}|\Psi_N(x_1,y_2,Z)|^2dx_1dy_2dZ
\\
=&\|\phi\|^2_{L^2(\bR^3)}\|\Psi_N\|^2_{L^2(\bR^{3N})}\,.
\ea
$$
Hence
$$
\ba
\left(\int_{\bR^3}|f(x_1,x_1+X,x_1)|dx_1\right)^2\le\left(\int_{\bR^{3N-3}}|\Psi_N(x_1,x_1+X,Z)|\Phi_N(x_1,Z)dZdx_1\right)^2
\\
\le\int_{\bR^{3N-3}}|\Psi_N(x_1,x_1+X,Z)|^2dZdx_1\int_{\bR^{3N-3}}\Phi_N(x_1,Z)^2dZdx_1
\\
\le\int_{\bR^{3N-3}}|\Psi_N(x_1,x_1+X,Z)|^2dZdx_1\|\phi\|^2_{L^2(\bR^3)}\|\Psi_N\|^2_{L^2(\bR^{3N})}\,,
\ea
$$
so that
$$
\ba
\int_{\bR^3}\left(\int_{\bR^3}|f(x_1,x_1+X,x_1)|dx_1\right)^2dX
\\
\le\|\phi\|^2_{L^2(\bR^3)}\|\Psi_N\|^2_{L^2(\bR^{3N})}\int_{\bR^3}\int_{\bR^{3N-3}}|\Psi_N(x_1,x_1+X,Z)|^2dZdx_1dX
\\
=\|\phi\|^2_{L^2(\bR^3)}\|\Psi_N\|^2_{L^2(\bR^{3N})}\int_{\bR^3}\int_{\bR^{3N-3}}|\Psi_N(x_1,x_2,Z)|^2dZdx_1dx_2
\\
=\|\phi\|^2_{L^2(\bR^3)}\|\Psi_N\|^4_{L^2(\bR^{3N})}&\,.
\ea
$$
Therefore
$$
\int_{\bR^3}|F(X)|^2dX\le\|\phi\|^2_{L^\infty(\bR^3)}\|\phi\|^2_{L^2(\bR^3)}\|\Psi_N\|^4_{L^2(\bR^{3N})}<\infty
$$
so that $\om\mapsto\hat F(-\om)$ belongs to $L^2(\bR^d)$ by Plancherel's theorem. Hence, for each $k\not=l\in\{1,\ldots,N\}$, one has
$$
\hat V\in L^1(\bR^d)+L^2(\bR^d)\implies\left\{\ba{}&\int_{\bR^3}|\hat V(\om)||\la\Psi_N|J_k(E_\om^*)J_l(E_\om|\phi\ra\la\phi|)|\Psi_N\ra| d\om<\infty\,,
\\&\int_{\bR^3}|\hat V(\om)||\la\Psi_N|J_l(|\phi\ra\la\phi|E_\om)J_k(E_\om^*)|\Psi_N\ra| d\om<\infty\,.\ea\right.
$$
Hence
$$
(t,\om)\mapsto\hat V(\om)\la\cU_N(t)\Psi_N^{in}|S_N[\phi](\om)|\cU_N(t)\Psi_N^{in}\ra\text{ belongs to }C_b(\bR_t,L^1(\bR^3_\om))\,.
$$
Since $S_N[\phi](\om)^*=-S_N[\phi](\om)\in\cL(\fH_N)$ for each $\om\in\bR^3$ and $\hat V$ is even because of \eqref{Veven}, the formula
$$
\ba
\la\Psi_N^{in}|\cC(V,\cM_N(t),\cM_N(t))(|\phi\ra\la\phi|)|\Psi_N^{in}\ra
\\
:=\tfrac1{(2\pi)^3}\int_{\bR^3}\hat V(\om)\la\cU_N(t)^*\Psi_N^{in}|S_N[\phi](\om)|\cU_N(t)^*\Psi_N^{in}\ra d\om
\ea
$$
defines 
$$
\cC(V,\cM_N(t),\cM_N(t))(|\phi\ra\la\phi|)=-\cC(V,\cM_N(t),\cM_N(t))(|\phi\ra\la\phi|)^*\in\cL(\fH_N)
$$
for each $t\in\bR$ by polarization, and the function
$$
t\mapsto\cC(V,\cM_N(t),\cM_N(t))
$$
is bounded on $\bR$ with values in $\cL(\fH_N)$ for the norm topology, and continuous on $\bR$ with values in $\cL(\fH_N)$ endowed with the weak operator topology, and therefore for the ultraweak topology (since the weak operator 
and the ultraweak topologies coincide on norm bounded subsets of $\cL(\fH_N)$).
\end{proof}

\bigskip
\noindent
\textbf{Remark.} In the sequel, we shall also need to consider terms of the form
$$
\ba
(I):=&\tfrac1{(2\pi)^3}\int_{\bR^3}\hat V(\om)\la\Psi_N|J_1(E_\om^*|\phi\ra\la\phi|)J_2(E_\om|\phi\ra\la\phi|)|\Psi_N\ra d\om
\\
(II):=&\tfrac1{(2\pi)^3}\int_{\bR^3}\hat V(\om)\la\Psi_N|J_1(|\phi\ra\la\phi|E_\om^*)J_2(E_\om|\phi\ra\la\phi|)|\Psi_N\ra d\om
\\
(III):=&\tfrac1{(2\pi)^3}\int_{\bR^3}\hat V(\om)\la\Psi_N|J_1(|\phi\ra\la\phi|E_\om^*|\phi\ra\la\phi|)J_2(AE_\om B)|\Psi_N\ra d\om
\ea
$$
where $A,B\in\cL(\fH)$. 

The term (III) is the easiest of all. Indeed, 
$$
\ba
(III)=&\tfrac1{(2\pi)^3}\int_{\bR^3}\hat V(\om)\widehat{|\phi|^2}(\om)\la\Psi_N|J_1(|\phi\ra\la\phi|)J_2(AE_\om B)|\Psi_N\ra d\om
\\
=&\la\Psi_N|J_1(|\phi\ra\la\phi|)J_2(A(V\star|\phi|^2)B)|\Psi_N\ra 
\ea
$$
and since $V\in L^2(\bR^3)+C_b(\bR^3)$ while $\phi\in L^1\cap L^\infty(\bR^3)$, one has $V\star|\phi|^2\in C_b(\bR^3)$, so that $A(V\star|\phi|^2)B\in\cL(\fH)$.

The terms (I) and (II) are slightly more delicate, but can be treated by the same method already used in the proof of the theorem above. First, 
$$
\la\Psi_N|J_1(E_\om^*|\phi\ra\la\phi|)J_2(E_\om|\phi\ra\la\phi|)|\Psi_N\ra=\hat F_1(\om)\,,
$$
with
$$
\ba
F_1(Y):=&\int_{\bR^{3N-6}}A_1(Y,Z)A_2(Z)dZ\,,
\\
A_1(Y,Z):=&\int_{\bR^3}\phi(X+\tfrac{Y}2)\phi(X-\tfrac{Y}2)\overline{\Psi_N(X+\tfrac{Y}2,X-\tfrac{Y}2,Z)}dX\,,
\\
A_2(Z):=&\int_{\bR^6}\!\Psi_N(y_1,y_2,Z)\overline{\phi(y_1)\phi(y_2)}dy_1dy_2\,,
\ea
$$
so that
$$
(I)=\tfrac1{(2\pi)^3}\int_{\bR^3}\hat V(\om)\hat F_1(\om)d\om\,.
$$
Then
$$
\ba
\left(\int_{\bR^3}|F_1(Y)|dY\right)^2
\\
\le\|A_2\|^2_{L^2(\bR^{3N-6})}\int_{\bR^{3N-6}}\left(\int_{\bR^3}|A_1(Y,Z)|dY\right)^2dZ\le\|A_2\|^2_{L^2(\bR^{3N-6})}
\\
\times\int_{\bR^{3N-6}}\left(\int_{\bR^6}|\phi(X+\tfrac{Y}2)||\phi(X-\tfrac{Y}2)||\Psi_N(X+\tfrac{Y}2,X-\tfrac{Y}2,Z)|dXdY\right)^2dZ
\\
\le\|A_2\|^2_{L^2(\bR^{3N-6})}\int_{\bR^6}|\phi(X+\tfrac{Y}2)|^2|\phi(X-\tfrac{Y}2)|^2dXdY
\\
\times\int_{\bR^{3N}}|\Psi_N(X+\tfrac{Y}2,X-\tfrac{Y}2,Z)|^2dXdYdZ
\\
=\|A_2\|^2_{L^2(\bR^{3N-6})}\int_{\bR^6}|\phi(x_1)|^2|\phi(x_2)|^2dx_1dx_2
\\
\times\int_{\bR^{3N}}|\Psi_N(x_1,x_2,Z)|^2dx_1dx_2dZ
\\
=\|A_2\|^2_{L^2(\bR^{3N-6})}\|\phi\|_{L^2(\bR^3)}^4\|\Psi_N\|_{\fH_N}^2&\,.
\ea
$$
Besides
$$
\ba
\|A_2\|^2_{L^2(\bR^{3N-6})}\le&\int_{\bR^6}|\phi(y_1)|^2|\phi(y_2)|^2dy_1dy_2\int_{\bR^{3N-6}}|\Psi_N(y_1,y_2,Z)|^2dy_1dy_2dZ
\\
=&\|\phi\|_{L^2(\bR^3)}^4\|\Psi_N\|_{\fH_N}^2\,,
\ea
$$
so that
$$
\|F_1\|_{L^1(\bR^3)}\le\|\phi\|_{L^2(\bR^3)}^4\|\Psi_N\|_{\fH_N}^2<\infty\,.
$$
On the other hand
$$
\int_{\bR^3}|F_1(Y)|^2dY\le\|A_2\|^2_{L^2(\bR^{3N-6})}\|A_1\|^2_{L^2(\bR^{3N-3})}\,,
$$
where
$$
\ba
\|A_1\|^2_{L^2(\bR^{3N-3})}\le&\sup_{Y\in\bR^3}\int_{\bR^3}|\phi(X+\tfrac{Y}2)|^2|\phi(X-\tfrac{Y}2)|^2dX
\\
&\times\int_{\bR^{3N}}|\Psi_N(X+\tfrac{Y}2,X-\tfrac{Y}2,Z)|^2dXdYdZ
\\
\le&\|\phi\|_{L^4(\bR^3)}^4\|\Psi_N\|_{\fH_N}^2\,,
\ea
$$
so that
$$
\int_{\bR^3}|F_1(Y)|^2dY\le\|\phi\|_{L^4(\bR^3)}^4\|\phi\|_{L^2(\bR^3)}^4\|\Psi_N\|_{\fH_N}^4<\infty\,.
$$
Thus, we have proved that $F_1\in L^1\cap L^2(\bR^d)$, and since $\hat V\in L^2(\bR^d)+L^1(\bR^d)$, the product $\hat V\hat F\in L^1(\bR^d)$, which leads to a definition of (I).

The case of (II) is essentially similar. Observe that
$$
\la\Psi_N|J_1(|\phi\ra\la\phi|E_\om^*)J_2(E_\om|\phi\ra\la\phi|)|\Psi_N\ra=\int_{\bR^{3N-6}}\hat F_2(\om,Z)\hat F_3(\om,Z)dZ\,,
$$
where
$$
\ba
F_2(y_1,Z):=&\overline{\phi(y_1)}\int_{\bR^3}\overline{\phi(y_2)}\Psi_N(y_1,y_2,Z)dy_2
\\
F_3(x_2,Z):=&\phi(x_2)\int_{\bR^3}\phi(x_1)\overline{\Psi_N(x_1,x_2,Z)}dx_1\,.
\ea
$$
And
$$
(II):=\tfrac1{(2\pi)^3}\int_{\bR^3}\hat V(\om)\left(\int_{\bR^{3N-6}}\hat F_2(\om,Z)\hat F_3(\om,Z)dZ\right)d\om\,.
$$

Observe that
$$
\ba
\left(\int|F_2(y_1,Z)|dy_1\right)^2\le&\left(\int_{\bR^6}|\phi(y_1)||\phi(y_2)||\Psi_N(y_1,y_2,Z)|dy_1dy_2\right)^2
\\
\le&\|\phi\|_{L^2(\bR^3)}^4\int_{\bR^6}|\Psi_N(y_1,y_2,Z)|^2dy_1dy_2\,,
\ea
$$
so that
$$
\ba
\sup_{\om\in\bR^3}\left|\int_{\bR^{3N-6}}\hat F_2(\om,Z)\hat F_3(\om,Z)dZ\right|^2
\\
\le\left(\int_{\bR^{3N-6}}\sup_{\om\in\bR^3}|\hat F_2(\om,Z)|\sup_{\om\in\bR^3}|\hat F_3(\om,Z)|dZ\right)^2
\\
\le\int_{\bR^{3N-6}}\sup_{\om\in\bR^3}|\hat F_2(\om,Z)|^2dZ\int_{\bR^{3N-6}}\sup_{\om\in\bR^3}|\hat F_3(\om,Z)|^2dZ
\\
\le\int_{\bR^{3N-6}}\left(\int|F_2(y_1,Z)|dy_1\right)^2dZ\int_{\bR^{3N-6}}\left(\int|F_3(x_2,Z)|dx_2\right)^2dZ
\\
\le\|\phi\|_{L^2(\bR^3)}^8\|\Psi_N\|_{\fH_N}^4&\,,
\ea
$$
while
$$
\int_{\bR^{3N-6}}\left(\int|F_2(y_1,Z)|dy_1\right)^2dZ\le\|\phi\|_{L^2(\bR^3)}^4\|\Psi_N\|_{\fH_N}^2\,,
$$
with a similar conclusion for $F_3$. On the other hand
$$
\ba
\int_{\bR^3}|F_2(y_1,Z)|^2dy_1\le\int_{\bR^3}|\phi(y_1)|^2\left(\int_{\bR^3}\overline{\phi(y_2)}\Psi_N(y_1,y_2,Z)dy_2\right)^2dy_1
\\
\le\|\phi\|_{L^2(\bR^3)}^2\int_{\bR^3}|\phi(y_1)|^2\left(\int_{\bR^3}|\Psi_N(y_1,y_2,Z)|^2dy_2\right)dy_1
\\
\le\|\phi\|^2_{L^2(\bR^3)}\|\phi\|^2_{L^\infty(\bR^3)}\int_{\bR^6}|\Psi_N(y_1,y_2,Z)|^2dy_1dy_2&\,,
\ea
$$
so that
$$
\ba
\int_{\bR^3}\left|\int_{\bR^{3N-6}}\hat F_2(\om,Z)\hat F_3(\om,Z)dZ\right|^2d\om
\\
\le\int_{\bR^3}\int_{\bR^{3N-6}}|\hat F_2(\om,Z)|^2\left(\int_{\bR^{3N-6}}|\hat F_3(\om,Z)|^2dZ\right)dZd\om
\\
\le\sup_{\om\in\bR^3}\int_{\bR^{3N-6}}|\hat F_3(\om,Z)|^2dZ\int_{\bR^{3N-6}}\int_{\bR^3}|\hat F_2(\om,Z)|^2dZd\om
\\
\le\int_{\bR^{3N-6}}\sup_{\om\in\bR^3}|\hat F_3(\om,Z)|^2dZ\int_{\bR^{3N-6}}(2\pi)^3\left(\int_{\bR^3}|F_2(y_1,Z)|^2dy_1\right)dZ
\\
\le(2\pi)^3\int_{\bR^{3N-6}}\left(\int_{\bR^3}|\hat F_3(x_2,Z)|dx_2\right)^2dZ\int_{\bR^{3N-6}}\int_{\bR^3}|F_2(y_1,Z)|^2dy_1dZ
\\
\le(2\pi)^3\|\phi\|_{L^2(\bR^3)}^6\|\phi\|_{L^\infty(\bR^3)}^2\|\Psi_N\|_{\fH_N}^4&\,.
\ea
$$
Therefore the map
$$
\om\mapsto\int_{\bR^{3N-6}}\hat F_2(\om,Z)\hat F_3(\om,Z)dZ
$$
belongs to $L^2\cap L^\infty(\bR^3)$. Since $\hat V\in L^2(\bR^3)+L^1(\bR^3)$, this implies that
$$
\om\mapsto\hat V\int_{\bR^{3N-6}}\hat F_2(\om,Z)\hat F_3(\om,Z)dZ
$$
belongs to $L^1(\bR^3)$, thereby leading to a definition of (II).

%%%%%%%%%%%%%%%%%%%%%%%%%%%%%%%%%%%%%%%%%%%%%%%%%%%%%%%%%%%%%%%%%%%%%%%%%%%%%%%%%%%%%%%%%%%%%%%%%%%%%%%%%%%%%%%%%%%%%%%%%

\section{An Operator Inequality. Application to the Mean-Field Limit}\lb{S-Main}

%%%%%%%%%%%%%%%%%%%%%%%%%%%%%%%%%%%%%%%%%%%%%%%%%%%%%%%%%%%%%%%%%%%%%%%%%%%%%%%%%%%%%%%%%%%%%%%%%%%%%%%%%%%%%%%%%%%%%%%%%

First consider the Cauchy problem for the time dependent Hartree equation \eqref{TDHWf}. Assuming that the potential $V$ satisfies \eqref{KatoCondV} and \eqref{Veven}, for each $\phi^{in}\in H^2(\bR^3)$, there exists a unique solution 
$\phi\in C(\bR,H^2(\bR^3))$ of \eqref{TDHWf} by Theorems 1.4 and 1.3 of \cite{Hirata}.

Pickl's key idea in his proof of the mean-field limit in quantum mechanics is to consider the following functional (see Definition 2.2 and formula (6) in \cite{Pickl1}, with the choice $n(k):=k/N$, in the notation of \cite{Pickl1}):
$$
\a_N(\Psi_N,\psi):=\La\Psi_N\Big|\frac1N\sum_{k=1}^NJ_k(I_\fH-|\psi\ra\la\psi|)\Big|\Psi_N\Ra=\la\Psi_N|\cM_N^{in}(I_\fH-|\psi\ra\la\psi|)|\Psi_N\ra
$$
for all $\Psi_N\in\fH_N$ and $\psi\in\fH$. 

Assuming that $\psi\equiv\psi(t,x)$ is a solution of \eqref{TDHWf} while $\Psi_N(t,\cdot):=\cU_N(t)\Psi_N^{in}$, Pickl studies in section 2.1 of \cite{Pickl1} the time-dependent function $t\mapsto\a_N(\Psi_N(t,\cdot),\psi(t,\cdot))$, and proves
that it satisfies some Gronwall inequality.

Observe first that Pickl's functional $\a_N(\Psi_N(t,\cdot),\psi(t,\cdot))$ can be recast in terms of the quantum Klimontovich solution $\cM_N(t)$ as follows
\be\lb{PicklQKlim}
\ba
\a_N(\cU_N(t)\Psi^{in}_N,\psi(t,\cdot))=&\la\cU_N(t)\Psi^{in}_N|\cM_N^{in}(I_\fH-|\psi(t,\cdot)\ra\la\psi(t,\cdot)|)\cU_N(t)|\Psi^{in}_N\ra
\\
=&\la\Psi^{in}_N|\cU_N(t)\left(\cM_N^{in}(I_\fH-|\psi(t,\cdot)\ra\la\psi(t,\cdot)|)\right)\cU_N(t)|\Psi^{in}_N\ra
\\
=&\la\Psi_N^{in}|\cM_N(t)(I_\fH-|\psi(t,\cdot)\ra\la\psi(t,\cdot)|)|\Psi^{in}_N\ra\,.
\ea
\ee
This identity suggests therefore to deduce from \eqref{KlimEq} and \eqref{TDHWf} the expression of 
$$
\tfrac{d}{dt}\cM_N(t)(I_\fH-|\psi(t,\cdot)\ra\la\psi(t,\cdot)|)
$$
in terms of the interaction operator $\cC$ defined in \eqref{DefC}.

This is done in the first part of the next theorem, which is our second main result in this paper. 

\begin{Thm}\lb{T-OpIneq}
Assume that the (real-valued) interaction potential $V$, viewed as an (unbounded) multiplication operator acting on $\fH:=L^2(\bR^3)$, satisfies the parity condition \eqref{Veven} and  \eqref{CondV1}.

Let $\psi^{in}\in H^2(\bR^3)$ satisfy $\|\psi^{in}\|_{\fH}=1$, let $\psi$ be the solution of the Cauchy problem \eqref{TDHWf} for the time-dependent Hartree equation, and set
\be\lb{R=}
R(t):=|\psi(t,\cdot)\ra\la\psi(t,\cdot)|\,,\quad\text{ and }\quad P(t):=I_\fH-R(t)\,.
\ee
Then

\smallskip
\noindent
(1) the $N$-body quantum Klimontovich solution $t\mapsto\cM_N(t)$ satisfies 
$$
i\hb\d_t(\cM_N(t)(P(t)))=\cC(V,\cM_N(t)-\cR(t),\cM_N(t))(R(t))\,,
$$
where
$$
\cR(t)A:=\la\psi(t,\cdot)|A|\psi(t,\cdot)\ra I_{\fH_N}=\Tr_{\fH}(R(t)A)I_{\fH_N}\,;
$$
(2) the operator $\cC(V,\cM_N(t)-\cR(t),\cM_N(t))(P(t))$ is skew-adjoint on $\fH_N$ and satisfies the operator inequality
$$
\pm i\cC(V,\cM_N(t)-\cR(t),\cM_N(t))(R(t))\le 6L(t)\left(\cM_N(t)(P(t))+\tfrac{2}NI_{\fH_N}\right)\,,
$$
where\footnote{We recall that, if $E,F$ are Banach spaces
$$
\|v\|_{E\cap F}:=\max(\|v\|_E,\|v\|_F)\,,
$$
and
$$
\|f\|_{L^p(\bR^d)+L^q(\bR^d)}=\inf\{\|f_1\|_{L^p(\bR^d)}\!+\!\|f_2\|_{L^q(\bR^d)}\text{ s.t. }f\!=\!f_1\!+\!f_2\text{ with }f_1\!\in\! L^p(\bR^d),\,f_2\!\in\! L^q(\bR^d)\}\,.
$$}
\be\lb{DefL}
L(t):=2\max(1,C_S)\|V\|_{L^2(\bR^3)+L^\infty(\bR^3)}\|\psi(t,\cdot)\|_{H^2(\bR^3)}\,,
\ee
and where $C_S$ is the norm of the Sobolev embedding $H^2(\bR^3)\subset L^\infty(\bR^3)$.
\end{Thm}

\smallskip
The operator inequality for quantum Klimontovich solutions in the case of potentials with Coulomb type singularity obtained in part (2) of Theorem \ref{T-OpIneq} can be thought of as the reformulation of Pickl's argument in terms of 
the quantum Klimontovich solution $\cM_N(t)$. 

Indeed, we deduce from parts (1) and (2) in Theorem \ref{T-OpIneq} the operator inequality
\be\lb{GronwOp}
\tfrac{d}{dt}\cM_N(t)(P(t))\le\tfrac{6L(t)}\hbar\left(\cM_N(t)(P(t))+\tfrac{2}NI_{\fH_N}\right)\,.
\ee
Then, evaluating both sides of this inequality on the initial $N$-particle state $\Psi_N^{in}$, and taking into account the identity \eqref{PicklQKlim} leads to the Gronwall inequality
$$
\tfrac{d}{dt}\a_N(\cU_N(t)\Psi^{in}_N,\psi(t,\cdot))\le\tfrac{6L(t)}\hbar\left(\a_N(\cU_N(t)\Psi^{in}_N,\psi(t,\cdot))+\tfrac{2}N\right)
$$
satisfied by Pickl's functional $\a_N(\cU_N(t)\Psi^{in}_N,\psi(t,\cdot))$. This last inequality corresponds to inequality (11) and Lemma 3.2 in \cite{Pickl1}.

\smallskip
In the sequel, we shall denote by $\cL^p(\fH)$ for $p\ge 1$ the Schatten two-sided ideal of $\cL(\fH)$ consisting of operators $T$ such that
$$
\|T\|_p:=\left(\Tr_\fH((T^*T)^{p/2})\right)^{1/p}<\infty\,.
$$
In particular $\cL^1(\fH)$ is set of trace-class operators on $\fH$ and $\|\cdot\|_1$ the trace norm, while $\cL^2(\fH)$ is set of Hilbert-Schmidt operators on $\fH$ and $\|\cdot\|_2$ the Hilbert-Schmidt norm.

\begin{Cor}\lb{C-MF}
Under the same assumptions and with the same notations as in Theorem \ref{T-OpIneq}, consider the $N$-body wave function $\Psi_N(t,\cdot):=\cU_N(t)(\psi^{in})^{\otimes N}$, and the $N$-body density operator 
$F_N(t):=|\Psi_N(t,\cdot)\ra\la\Psi_N(t,\cdot)|$. For each $m=1,\ldots,N$, the $m$-particle reduced density operator $F_{N:m}(t)$, defined by the identity
$$
\Tr_{\fH_m}(F_{N:m}(t)A_1\otimes\ldots\otimes A_m)=\la\Psi_N(t,\cdot)|A_1\otimes\ldots\otimes A_m\otimes I_{\fH_{N-m}}|\Psi_N(t,\cdot)\ra
$$
for all $A_1,\ldots,A_m\in\cL(\fH)$, satisfies
$$
\|F_{N:m}(t)-R(t)^{\otimes m}\|_1\le 4\sqrt{\frac{m}N}\exp\left(\tfrac{3}\hb\int_0^tL(s)ds\right)\,,
$$
with $L$ given by \eqref{DefL}.
\end{Cor}

\smallskip
Let us briefly indicate how one arrives at the operator inequality in part (2) of Theorem \ref{T-OpIneq}. Let $\Lambda_1,\Lambda_2\in\cL(\cL(\fH),\cL(\fH_N))$ be such that
\be\lb{CondLa1La2}
\om\mapsto\la\Psi_N|\Lambda_1(E_\om^*)\Lambda_2(E_\om)|\Psi_N\ra\text{belongs to }L^1\cap L^2(\bR^3)
\ee
for all $\Psi_N\in\fH_N$. For all $V$ satisfying \eqref{Veven} and \eqref{CondV1}, define $\cT(V,\L_1,\L_2)\in\cL(\fH_N)$ by polarization of the formula
$$
\la\Psi_N|\cT(V,\L_1,\L_2)|\Psi_N\ra:=\tfrac1{(2\pi)^d}\int_{\bR^d}\hat V(\om)\la\Psi_N|\Lambda_1(E_\om^*)\Lambda_2(E_\om)|\Psi_N\ra d\om\,.
$$
In other words,
\be\lb{Def-cT}
\cT(V,\Lambda_1,\Lambda_2):=\tfrac1{(2\pi)^d}\int_{\bR^d}\hat V(\om)\Lambda_1(E_\om^*)\Lambda_2(E_\om)d\om
\ee
where the integral on the right hand is to be understood in the ultraweak sense (see footnote 3 on p. 1032 in \cite{FGTPaul}).

For each $A\in\cL(\fH)$, denote by $\Lambda_j(\bu A)$ and $\Lambda_j(A\bu)$ the linear maps
$$
\ba
\Lambda_j(\bu A):\,\cL(\fH)\ni B\mapsto\Lambda_j(BA)\in\cL(\fH_N)
\\
\Lambda_j(A\bu):\,\cL(\fH)\ni B\mapsto\Lambda_j(AB)\in\cL(\fH_N)
\ea
$$
respectively. If $A\in\cL(\fH)$ is such that $\Lambda_1,\Lambda_2(\bu A)$ and $\Lambda_2(A\bu),\Lambda_1$ satisfy \eqref{CondLa1La2}, then one has
\be\lb{C=T-T}
\cC(V,\Lambda_1,\Lambda_2)A=\cT(V,\Lambda_1,\Lambda_2(\bu A))-\cT(V,\Lambda_2(A\bu),\Lambda_1)\,.
\ee

\begin{Lem}\lb{L-SelfAd}
Let $\Lambda_1,\Lambda_2\in\cL(\cL(\fH),\cL(\fH_N))$ be $*$-homomorphisms, in other words
$$
\Lambda_j(A^*)=\Lambda_j(A)^*\,,\qquad j=1,2
$$
for all $A\in\cL(\fH)$. Assume that $\Lambda_1,\Lambda_2$ satisfy \eqref{CondLa1La2}. Then
$$
\cT(V,\Lambda_2,\Lambda_1)=\cT(V,\Lambda_1,\Lambda_2)^*\,.
$$
\end{Lem}

\begin{proof}
Indeed
$$
\ba
\cT(V,\Lambda_2,\Lambda_1)=&\tfrac1{(2\pi)d}\int_{\bR^d}\hat V(\om)\Lambda_2(E_\om^*)\Lambda_1(E_\om)d\om
\\
=&\tfrac1{(2\pi)d}\int_{\bR^d}\hat V(\om)\Lambda_2(E_\om)^*\Lambda_1(E^*_\om)^*d\om=\cT(V,\Lambda_1,\Lambda_2)^*
\ea
$$
where the first equality follows from the fact that $\Lambda_1$ and $\Lambda_2$ are *-homomorphisms, while the second equality uses the fact that $\hat V$ is real-valued, since $V$ is real-valued and even. 
\end{proof}

An easy consequence of \eqref{C=T-T} and of this lemma is that, for each $A=A^*\in\cL(\fH)$ such that $\Lambda_1,\Lambda_2\in\cL(\cL(\fH),\cL(\fH_N))$ are $*$-homomorphisms such that $\Lambda_1,\Lambda_2(\bu A)$ satisfy 
\eqref{CondLa1La2}, then 
\be\lb{cC=-cC*}
(\cC(V,\Lambda_1,\Lambda_2)A)^*=-\cC(V,\Lambda_1,\Lambda_2)A\,.
\ee

\smallskip
The key observations leading to Theorem \ref{T-OpIneq} are summarized in the two following lemmas. In the first of these two lemmas, the interaction operator is decomposed into a sum of four terms.

\begin{Lem}\lb{L-Decomp-cC}
Under the same assumptions and with the same notations as in Theorem \ref{T-OpIneq}, the interaction operator satisfies the identity
$$
\cC(V,\cM_N(t)-\cR(t),\cM_N(t))(R(t))=T_1+T_2+T_3+T_4\,,
$$
with
$$
\ba
T_1:=&\cT(V,\cM_N(t)(P(t)\bu P(t)),\cM_N(t)(P(t)\bu R(t))
\\
&-\cT(V,\cM_N(t)(R(t)\bu P(t)),\cM_N(t)(P(t)\bu P(t))\,,
\\
T_2:=&\cM_N(t)(R(t)V_{R(t)}P(t))\cM_N(t)(P(t))
\\
&-\cM_N(t)(P(t))\cM_N(t)(P(t)V_{R(t)}R(t))\,,
\\
T_3:=&\cT(V,\cM_N(t)(P(t)\bu R(t)),\cM_N(t)(P(t)\bu R(t))
\\
&-\cT(V,\cM_N(t)(R(t)\bu P(t)),\cM_N(t)(R(t)\bu P(t))\,,
\\
T_4:=&\tfrac1N\cM_N(t)[V_{R(t)},R(t)]\,.
\ea
$$
\end{Lem}

All the terms involved in this decomposition can be defined by the same method already used in the proof of Theorem \ref{T-DefC}. Indeed, one can check that all these terms involve only expressions of the type (I), (II) or (III) in the 
Remark following Theorem \ref{T-DefC}. This easy verification is left to the reader, and we shall henceforth consider this matter as settled by the detailed explanations concerning (I), (II) and (III) given in the previous section.

Each term in this decomposition satisfies an operator inequality involving only the operator norm of the ``mean-field squared potential'' $(V^2)_{R(t)}$, instead of the ``bare'' interaction potential $V$ itself.

\begin{Lem}\lb{L-IneqT1234}
Under the same assumptions and with the same notations as in Theorem \ref{T-OpIneq}, set
\be\lb{Defell}
\ell(t):=\|V^2\star|\psi(t,\cdot)|^2\|^\frac12\,.
\ee
Then
$$
\ba
\pm iT_1\le&2\ell(t)\left((1-\tfrac1N)\cM_N(t)(P(t))+\tfrac4NI_{\fH_N}\right)\,,
\\
\pm iT_2\le&2\ell(t)(\cM_N(t)(P(t))+\tfrac1NI_{\fH_N})\,,
\\
\pm iT_3\le&2\ell(t)((1-\tfrac1N)\cM_N(t)(P(t))+\tfrac1NI_{\fH_N})\,,
\\
\pm iT_4\le&\tfrac2N\ell(t)I_{\fH_N}\,.
\ea
$$
\end{Lem}

\bigskip
\noindent
\textbf{Remarks on $\ell(t)$ in \eqref{Defell} and $L(t)$ in \eqref{DefL}.}

\smallskip
\noindent
(1) If $V$ satisfies condition \eqref{CondV1} in Theorem \ref{T-DefC}, then $V\!\in L^2(\bR^3)\!+\!L^\infty(\bR^3)$, so that $V^2\!\in L^1(\bR^d)+L^\infty(\bR^d)$. Thus $(V^2)_{R(t)}$, which is the multiplication operator
by the function $V^2\star|\psi(t,\cdot)|^2$, satisfies
$$
\ba
\ell(t)^2:=\|V^2\star|\psi(t,\cdot)|^2\|_{L^\infty(\bR^3}\le&\|V^2\|_{L^1(\bR^3)+L^\infty(\bR^3)}\|\psi(t,\cdot)\|^2_{L^1\cap L^\infty(\bR^3)}
\\
\le&2\|V\|^2_{L^1(\bR^3)+L^\infty(\bR^3)}\max(1,\|\psi(t,\cdot)\|_{L^\infty(\bR^3)})^2
\\
\le&2C^2_S\|V\|^2_{L^1(\bR^3)+L^\infty(\bR^3)}\|\psi(t,\cdot)\|_{H^2(\bR^3)}^2
\ea
$$
where we recall that $C_S$ is the norm of the Sobolev embedding $H^2(\bR^3)\subset L^\infty(\bR^3)$.

\smallskip
\noindent
(2) If $V$ satisfies \eqref{KatoCondV}, then $\|V(I-\Dlt)^{-1}\|\le M$ for some positive constant $M$ (see the discussion in \S 5.3 of chapter V in \cite{Kato}, so that
$$
V^2\le M^2(I-\Dlt)^2\,.
$$
In this remark, we shall make a slightly more restrictive assumption, namely that $V^2$ satisfies
\be\lb{IneqV}
V^2\le C(I-\Dlt)\,.
\ee
In space dimension $d=3$, the Hardy inequality, which can be put in the form\footnote{To see that $4$ is optimal, minimize in $\a>0$ the expression
$$
\int_{\bR^3}\left|\grad u+\a\frac{x}{|x|^2}u\right|^2dx\,.
$$}
$$
\frac1{|x|^2}\le 4(-\Dlt)
$$
implies that the Coulomb potential satisfies the assumption above on $V$. If the potential $V$ satisfies the (operator) inequality \eqref{IneqV}, then
$$
\ba
0\le(V^2)_{R(t)}(x)=\int_{\bR^d}V^2(y)|\psi(t,x-y)|^2dy=\la\psi(t,x-\cdot)|V^2|\psi(t,x-\cdot)\ra
\\
\le C\la\psi(t,x-\cdot)|(I-\Dlt)|\psi(t,x-\cdot)\ra=C\|\psi(t,x-\cdot)\|_{L^2}^2+C\|\grad\psi(t,x-\cdot)\|_{L^2}^2
\\
=C\|\psi(t,\cdot)\|_{L^2}^2+C\|\grad\psi(t,\cdot)\|_{L^2}^2&\,.
\ea
$$
Thus, if $\psi\in C(\bR;H^1(\bR^d))$ is a solution of the Hartree equation, 
$$
\ell(t)\le\sqrt{C}\|\psi(t,\cdot)\|_{H^1(\bR^3)}\,.
$$
(3) A bound on $\ell(t)$ in terms of $\|\psi(t,\cdot)\|_{H^1(\bR^3)}$ instead of $\|\psi(t,\cdot)\|_{H^2(\bR^3)}$ is advantageous since the former quantity can be controlled rather explicitly by means of the conservation of energy 
for the Hartree equation \eqref{TDHWf}. This explicit control is useful in particular to assess the dependence in $\hbar$ of the convergence rate for the mean-field limit obtained in Corollary \eqref{C-MF}.

Clearly, the convergence rate for the quantum mean-field limit in Corollary \ref{C-MF} is not uniform in the semiclassical regime, in the first place because of the factor $3/\hb$ on the right hand side of the upper bound for 
$\|F_{N:m}(t)-R(t)^{\otimes m}\|_1$, which comes from the $i\hb\d_t$ part of the quantum dynamical equation. 

However, one should expect that the function $\ell(t)$, or at least the upper bound for $\ell(t)$ obtained in (2), grows at least as $1/\hb$, since it involves $\|\grad_x\psi(t,\cdot)\|_{L^2}$, expected to be of order $1/\hb$ for 
semiclassical wave functions $\psi$ (think for instance of a WKB wave function, or of a Schr\"odinger coherent state).

We shall discuss this issue by means of the conservation of energy satisfied by the Hartree solution $\psi$ (see formula (5.2) in \cite{Bove2}):
$$
\ba
\tfrac12\hb^2\|\grad\psi(t,\cdot)\|_{L^2}^2+\tfrac12\int_{\bR^d}|\psi(t,x)|^2(V\star|\psi(t,x)|^2)dx
\\
=\tfrac12\hb^2\|\grad\psi^{in}\|_{L^2}^2+\tfrac12\int_{\bR^d}|\psi^{in}(x)|^2(V\star|\psi^{in}(x)|^2)dx&\,.
\ea
$$
Observe that
\be\lb{BndMFpot}
|V\star|\psi(t,x)|^2|\le\|\psi(t,\cdot)\|_{L^2}(V^2\star|\psi(t,x)|^2)^{1/2}=\ell(t)\,,
\ee
so that
$$
\ba
\tfrac12\hb^2\|\grad\psi^{in}\|_{L^2}^2+\tfrac12\int_{\bR^d}|\psi^{in}(x)|^2(V\star|\psi^{in}(x)|^2)dx
\\
\le\tfrac12\hb^2\|\psi^{in}\|_{H^1}^2+\tfrac12\ell(t)\le\tfrac12\hb^2\|\psi^{in}\|_{H^1}^2+\tfrac12\sqrt{C}\|\psi^{in}\|_{H^1}&\,.
\ea
$$

Thus, if $V\ge 0$, or if $\hat V\ge 0$, one has
$$
\int_{\bR^d}|\psi(t,x)|^2(V\star|\psi(t,x)|^2)dx=\tfrac1{(2\pi)^d}\int_{\bR^d}\hat V(\om)|\cF(|\psi(t,\cdot)|^2)|^2(\om)d\om\ge 0
$$
(where $\cF$ designates the Fourier transform on $\bR^d$), so that the conservation of mass and energy for the Hartree solution implies that
$$
\hb^2\|\psi(t,\cdot)\|_{H^1}^2\le\hb^2\|\psi^{in}\|_{H^1}^2+\sqrt{C}\|\psi^{in}\|_{H^1}\,.
$$
In that case
$$
\ell(t)\le\tfrac1\hb\sqrt{C(\hb^2\|\psi^{in}\|_{H^1}^2+\sqrt{C}\|\psi^{in}\|_{H^1})}\,.
$$
Typical states used in the semiclassical regime (WKB or coherent states, for instance) satisfy $\hb\|\grad\psi^{in}\|_{L^2}=O(1)$. Thus, in that case
$$
\ell(t)\le\hb^{-3/2}\sqrt{C(\hb^3\|\psi^{in}\|_{H^1}^2+\sqrt{C}\hb\|\psi^{in}\|_{H^1})}=O(\hb^{-3/2})\,.
$$

Things become worse if the potential energy is a priori of indefinite sign. With \eqref{BndMFpot}, the energy conservation implies that
$$
\ba
\hb^2\|\psi(t,\cdot)\|_{H^1}^2\le&\hb^2\|\psi^{in}\|_{H^1}^2+\sqrt{C}\|\psi^{in}\|_{H^1}+\sqrt{C}\|\psi(t,\cdot)\|_{H^1}
\\
\le&\hb^2\|\psi^{in}\|_{H^1}^2+\sqrt{C}\|\psi^{in}\|_{H^1}+\tfrac{C}{2\hb^2}+\tfrac12\hb^2\|\psi(t,\cdot)\|^2_{H^1}\,,
\ea
$$
so that
$$
\hb^2\|\psi(t,\cdot)\|_{H^1}^2\le 2\left(\hb^2\|\psi^{in}\|_{H^1}^2+\sqrt{C}\|\psi^{in}\|_{H^1}+\tfrac{C}{2\hb^2}\right)\le3\hb^2\|\psi^{in}\|_{H^1}^2+2\tfrac{C}{\hb^2}\,,
$$
and thus
$$
\ell(t)\le\hb^{-2}\sqrt{C(3\hb^4\|\psi^{in}\|_{H^1}^2+2C)}=O(\hb^{-2})\,.
$$

Therefore, the exponential amplifying factor in Corollary \ref{C-MF} is $\exp(Kt/\hb^{5/2})$ in the first case, and $\exp(Kt/\hb^3)$ in the second. These elementary remarks suggest that Pickl's clever method for proving the quantum mean-field
limit with singular potentials including the Coulomb potential (see \cite{Pickl1,KPickl}) is not expected to give uniform convergence rates (as in \cite{FGMouPaul,FGTPaul} in the case of regular interaction potentials) for the mean field limit in 
the semiclassical regime.

%%%%%%%%%%%%%%%%%%%%%%%%%%%%%%%%%%%%%%%%%%%%%%%%%%%%%%%%%%%%%%%%%%%%%%%%%%%%%%%%%%%%%%%%%%%%%%%%%%%%%%%%%%%%%%%%%%%%%%%%%

\section{Proof of part (1) in Theorem \ref{T-OpIneq}}

%%%%%%%%%%%%%%%%%%%%%%%%%%%%%%%%%%%%%%%%%%%%%%%%%%%%%%%%%%%%%%%%%%%%%%%%%%%%%%%%%%%%%%%%%%%%%%%%%%%%%%%%%%%%%%%%%%%%%%%%%

For each $\si\in\fS_N$ and each $\Psi_N\in\fH_N$, set
$$
(U_\si\Psi_N)(x_1,\ldots,x_N)=\Psi_N(x_{\si^{-1}(1)},\ldots,x_{\si^{-1}(N)})\,.
$$

Since $\psi(t,\cdot)\in H^2(\bR^3)$, the commutator $[\Dlt,R(t)]$ is a bounded operator on $\fH$. According to formula (25) in \cite{FGTPaul}, denoting by $V_{kl}$ the multiplication operator
\be\lb{DefVkl}
(V_{kl}\Psi_N)(x_1,\ldots,x_N)=V(x_k-x_l)\Psi_N(x_1,\ldots,x_N)\,,
\ee
one has
\be\lb{Csq25}
\ba
\Tr_{\fH_N}&((i\hb\d_t\cM_N(t)-\ad^*(-\tfrac12\hb^2\Dlt)\cM_N(t))(P(t))F_N)
\\
&=-\Tr_{\fH_N}(\tfrac{N-1}N([V_{12},J_1P(t)])F_N)=\Tr_{\fH_N}(\tfrac{N-1}N([V_{12},J_1R(t)])F_N)
\ea
\ee
for all $F_N\in\cL(\fH_N)$ such that
\be\lb{SymDens}
F_N=F_N^*\ge 0\,,\quad\Tr_{\fH_N}(F_N)=1\,,\quad\text{ and }\quad U_\si F_NU_\si^*=F_N\text{ for all }\si\in\fS_N\,.
\ee

The core result in the proof of Theorem \ref{T-DefC} is that the function
$$
\om\mapsto\la\Psi_N|J_k([E_\om,R(t)])J_l(E_\om^*)|\Psi_N\ra\in L^2\cap L^\infty(\bR^3)
$$
for each $k\not=l\in\{1,\ldots,N\}$. Since $\hat V\in L^1(\bR^3)+L^2(\bR^3)$, this has led us to define
$$
\la\Psi_N|[V_{kl},J_kR(t)]|\Psi_N\ra:=\tfrac1{(2\pi)^3}\int_{\bR^3}\hat V(\om)\la\Psi_N|J_k([E_\om,R(t)])J_l(E_\om^*)|\Psi_N\ra d\om\,,
$$
and more generally, using a spectral decomposition of the trace-class operator $F_N$, 
$$
\Tr_{\fH_N}([V_{kl},J_kR(t)]F_N):=\tfrac1{(2\pi)^3}\int_{\bR^3}\hat V(\om)\Tr_{\fH_N}(J_k([E_\om,R(t)])J_l(E_\om^*)F_N)d\om
$$
with
$$
\om\mapsto\Tr_{\fH_N}(|J_k([E_\om,R(t)])J_l(E_\om^*)F_N)\in L^2\cap L^\infty(\bR^3)\,.
$$

Since $U_\si F_NU_\si^*=F_N$ for all $\si\in\fS_N$, for each $m\not=n\in\{1,\ldots,N\}$, one has
$$
\ba
\Tr_{\fH_N}(\tfrac{N-1}N([V_{12},J_1R(t)])F_N)=\Tr_{\fH_N}(\tfrac{N-1}N([V_{mn},J_mR(t)])F_N)
\\
=\tfrac1{N^2}\sum_{1\le k\not=l\le N}\tfrac1{(2\pi)^3}\int_{\bR^3}\hat V(\om)\Tr_{\fH_N}(J_k([E_\om,R(t)])J_l(E_\om^*)F_N)d\om
\\
=\tfrac1{(2\pi)^3}\int_{\bR^3}\hat V(\om)\Tr_{\fH_N}(S_N[\psi(t,\cdot)](\om)F_N)d\om&\,.
\ea
$$

With the definition of $\cC$ in Theorem \ref{T-DefC}, we conclude that the operator 
$$
S_N=(i\hb\d_t\cM_N(t)-\ad^*(-\tfrac12\hb^2\Dlt)\cM_N(t))(P(t))-\cC(V,\cM_N(t),\cM_N(t))(R(t))
$$
satisfies
$$
\Tr_{\fH_N}(S_NF_N)=0
$$
for each operator $F_N\in\cL(\fH_N)$ satisfying \eqref{SymDens}. One easily checks that
$$
U^*_\si S_NU_\si=S_N\quad\text{ for all }\si\in\fS_N\,.
$$
Let $D_N\in\cL(\fH_N)$ be a density operator on $\fH_N$, i.e.
\be\lb{Dens}
D_N=D_N^*\ge 0\quad\text{ and }\quad\Tr_{\fH_N}(D_N)=1\,.
\ee
Obviously
$$
F_N:=\tfrac1{N!}\sum_{\si\in\fS_N}U_\si D_NU_\si^*
$$
satisfies \eqref{SymDens}, so that
$$
0=\Tr_{\fH_N}(S_NF_N)=\tfrac1{N!}\sum_{\si\in\fS_N}\Tr_{\fH_N}(U_\si^*S_NU_\si D_N)=\Tr_{\fH_N}(S_ND_N)
$$
for all $D_N\in\cL(\fH_N)$ satisfying \eqref{Dens}. Since any trace-class operator on $\fH_N$ is a linear combination of $4$ density operators, we conclude that
$$
\Tr_{\fH_N}(S_NT_N)=0\qquad\text{ for all }T_N\in\cL^1(\fH_N)\,,
$$
so that
\be\lb{SN=0}
S_N=0\,.
\ee

On the other hand
$$
\ba
\cM_N(t)(i\hb\d_tP(t))=\cM_N(t)([-\tfrac12\hb^2\Dlt+V\star|\psi(t,\cdot)|^2,P(t)])
\\
=\cM_N(t)([-\tfrac12\hb^2\Dlt,P(t)])-\cM_N(t)([V\star|\psi(t,\cdot)|^2,R(t)])
\ea
$$
so that
\be\lb{dtMNtPt}
\ba
i\hb\d_t(\cM_N(t)(P(t)))=&\ad^*(-\tfrac12\hb^2\Dlt)\cM_N(t)(P(t))+\cC(V,\cM_N(t),\cM_N(t))(R(t))
\\
&+\cM_N(t)([-\tfrac12\hb^2\Dlt,P(t)])-\cM_N(t)([V\star|\psi(t,\cdot)|^2,R(t)])
\\
=&\cC(V,\cM_N(t),\cM_N(t))(R(t))-\cM_N(t)([V\star|\psi(t,\cdot)|^2,R(t)])\,.
\ea
\ee

Finally, by condition \eqref{CondV1} on $V$, one has
$$
\psi(t,\cdot)\in H^2(\bR^3)\subset L^2\cap L^4(\bR^3)\implies V\star|\psi(t,\cdot)|^2\in\cF L^1(\bR^3)
$$
so that
\be\lb{VR=}
\ba
V\star|\psi(t,\cdot)|^2=&\tfrac1{(2\pi)^3}\int_{\bR^3}\hat V(\om)\cF(|\psi(t,\cdot)|^2)(\om)E_\om d\om
\\
=&\tfrac1{(2\pi)^3}\int_{\bR^3}\hat V(\om)\la\psi(t,\cdot)|E_\om^*|\psi(t,\cdot)\ra E_\om d\om
\\
=&\tfrac1{(2\pi)^3}\int_{\bR^3}\hat V(\om)\cR(t)(E_\om^*)E_\om d\om\,.
\ea
\ee
Hence
\be\lb{MNVR=}
\ba
\cM_N(t)([V\star|\psi(t,\cdot)|^2,R(t)])=&\tfrac1{(2\pi)^3}\int_{\bR^3}\hat V(\om)\cR(t)(E_\om^*)\cM_N(t)[E_\om,R(t)] d\om
\\
=&\cC(V,\cR(t),\cM_N(t))(R(t))
\ea
\ee
so that, returning to \eqref{dtMNtPt}, one arrives at the equality
$$
i\hb\d_t(\cM_N(t)(P(t)))=\cC(V,\cM_N(t),\cM_N(t))(R(t))-\cC(V,\cR(t),\cM_N(t))(R(t))\,,
$$
which proves part (1) in Theorem \ref{T-OpIneq}.

\smallskip
\noindent
\textbf{Remark.} In \cite{FGTPaul}, the equality 
$$
i\hb\d_t\cM_N(t)(A)=\ad^*(-\tfrac12\hb^2\Dlt)\cM_N(t)(A)-\cC(V,\cM_N(t),\cM_N(t))(A)
$$
is proved for all $A\in\cL(\fH)$ such that $[\Dlt,A]\in\cL(\fH)$ assuming that $V\in\cF L^1(\bR^3)$. This argument cannot be used here since $V\notin\cF L^1(\bR^3)$. Besides, the definition of the operator $\cC(V,\cM_N(t),\cM_N(t))(R(t))$
in Theorem \ref{T-DefC} makes critical use of the fact that $R(t)=|\psi(t,\cdot)\ra\la\psi(t,\cdot)|$ with $\psi(t,\cdot)\in L^2\cap L^\infty(\bR^3)$. This is the reason for the rather lengthy justification of \eqref{SN=0} in this section.

%%%%%%%%%%%%%%%%%%%%%%%%%%%%%%%%%%%%%%%%%%%%%%%%%%%%%%%%%%%%%%%%%%%%%%%%%%%%%%%%%%%%%%%%%%%%%%%%%%%%%%%%%%%%%%%%%%%%%%%%%

\section{Proof of Lemma \ref{L-Decomp-cC}}

%%%%%%%%%%%%%%%%%%%%%%%%%%%%%%%%%%%%%%%%%%%%%%%%%%%%%%%%%%%%%%%%%%%%%%%%%%%%%%%%%%%%%%%%%%%%%%%%%%%%%%%%%%%%%%%%%%%%%%%%%

In the sequel, we seek to ``simplify'' the expression of the interaction operator 
$$
\cC(V,\cM_N(t)-\cR(t),\cM_N(t))(R(t))\,.
$$
This will lead to rather involved computations which do not seem much of a simplification. However, we shall see that the final result of these computations, reported in Lemma \ref{L-Decomp-cC}, although algebraically  more cumbersome, 
has better analytical properties. 

\subsection{A First Simplification}
%%%%%%%%%%%%%%%%%%%%%%%%%%%%%%%%%%%%%%%%%%%%%%%%%%%%%%%%%%%%%%%%%%%%%%%%%%%%%%%%%%%%%%%%%%%%%%%%%%%%%%%%%%%%%%%%%%%%%%%%%

First we decompose $E_\om R(t)$ and $R(t)E_\om$ in the terms $\cM_N(t)(E_\om R(t))$ and $\cM_N(t)(R(t)E_\om)$ as
$$
E_\om R(t)=P(t)E_\om R(t)+R(t)E_\om R(t)\,,
$$ 
and observe that
$$
\ba
\cC(V,\cM_N(t)-\cR(t),\cM_N(t))(R(t))
\\
=\tfrac1{(2\pi)^3}\int_{\bR^3}\hat V(\om)((\cM_N(t)-\cR(t))(E_\om^*)\cM_N(t)(P(t)E_\om R(t))
\\
-\cM_N(t)(R(t)E_\om P(t))(\cM_N(t)-\cR(t))(E_\om^*))d\om
\\
+\tfrac1{(2\pi)^3}\int_{\bR^3}\hat V(\om)[(\cM_N(t)-\cR(t))(E_\om^*),\cM_N(t)(R(t)E_\om R(t))]d\om&\,.
\ea
$$
All the terms in the right hand side of the equality above are either similar to the one considered in Theorem \ref{T-DefC}, or of the type denoted (III) in the Remark following Theorem \ref{T-DefC}.

An elementary computation shows that, for all $\om\in\bR^d$,
$$
\ba
{}[(\cM_N(t)-\cR(t))(E_\om^*),\cM_N(t)(R(t)E_\om R(t))]
\\
=[\cM_N(t)(E_\om^*),\cM_N(t)(R(t)E_\om R(t))]
\\
=\tfrac1N\cM_N(t)[E_\om^*,R(t)E_\om R(t)]&\,.
\ea
$$
Recall indeed that, for each $A,B\in\cL(\fH)$, one has
$$
[\cM_N(t)A,\cM_N(t)B]=\tfrac1N\cM_N(t)([A,B])
$$
--- see formula before (41) on p. 1041 in \cite{FGTPaul}. On the other hand
\be\lb{RComprEom}
R(t)E_\om R(t)=|\psi(t,\cdot)\ra\la\psi(t,\cdot)|E_\om|\psi(t,\cdot)\ra\la\psi(t,\cdot)|=\cF(|\psi(t,\cdot)|^2)(-\om)R(t)\,,
\ee
so that
$$
\ba
{}[(\cM_N(t)-\cR(t))(E_\om^*),\cM_N(t)(R(t)E_\om R(t))]
\\
=\tfrac1N\cF(|\psi(t,\cdot)|^2)(-\om)\cM_N(t)[E_\om^*,R(t)]&\,.
\ea
$$

Besides
$$
\ba
(\cM_N(t)-\cR(t))(E_\om^*)=&\cM_N(t)E_\om^*-\la\psi(t,\cdot)|E_\om^*\psi(t,\cdot)|\ra\,I_{\fH_N}
\\
=&\cM_N(t)E_\om^*-\cF(|\psi(t,\cdot)|^2)(\om)I_{\fH_N}
\\
=&\cM_N(t)(E_\om^*-\cF(|\psi(t,\cdot)|^2)(\om)I_{\fH})\,.
\ea
$$
Indeed
\be\lb{MN(I)}
\cM_N^{in}I_\fH =I_{\fH_N}\implies\cM_N(t)I_\fH=\cU_N(t)^*(\cM_N^{in}I_\fH)\cU_N(t)=I_{\fH_N}
\ee
where we recall that $\cU_N(t):=e^{-it\cH_N/\hb}$, while $\cH_N$ is the $N$-body Hamiltonian.

\smallskip
Therefore
$$
\ba
\cC(V,\cM_N(t)-\cR(t),\cM_N(t))(R(t))
\\
=\tfrac1{(2\pi)^3}\int_{\bR^3}\hat V(\om)((\cM_N(t)(E_\om^*-\cF(|\psi(t,\cdot)|^2)(\om)I_\fH)\cM_N(t)(P(t)E_\om R(t))
\\
-\cM_N(t)(R(t)E_\om P(t))(\cM_N(t)(E_\om^*-\cF(|\psi(t,\cdot)|^2)(\om)I_\fH))d\om
\\
+\tfrac1{(2\pi)^3}\int_{\bR^3}\hat V(\om)\cF(|\psi(t,\cdot)|^2)(\om)\tfrac1N\cM_N(t)[E_\om,R(t)]d\om&\,,
\ea
$$
in view of \eqref{Veven}. With the formula \eqref{MNVR=}, we conclude that
\be\lb{1stSimpli}
\ba
\cC(V,\cM_N(t)-\cR(t),\cM_N(t))R(t)
\\
=\tfrac1{(2\pi)^3}\int_{\bR^3}\hat V(\om)((\cM_N(t)(E_\om^*-\cF(|\psi(t,\cdot)|^2)(\om)I_\fH)\cM_N(t)(P(t)E_\om R(t))
\\
-\cM_N(t)(R(t)E_\om P(t))(\cM_N(t)(E_\om^*-\cF(|\psi(t,\cdot)|^2)(\om)I_\fH))d\om
\\
+\tfrac1N\cM_N(t)[V\star|\psi(t,\cdot)|^2,R(t)]&\,.
\ea
\ee

\subsection{A Second Simplification}
%%%%%%%%%%%%%%%%%%%%%%%%%%%%%%%%%%%%%%%%%%%%%%%%%%%%%%%%%%%%%%%%%%%%%%%%%%%%%%%%%%%%%%%%%%%%%%%%%%%%%%%%%%%%%%%%%%%%%%%%%

Next we decompose $E_\om^*$ in $\cM_N(t)(E_\om^*)$ as
$$
E_\om^*=P(t)E_\om^*P(t)+P(t)E_\om^*R(t)+R(t)E_\om^*P(t)+R(t)E_\om^*R(t)\,.
$$
The identity \eqref{RComprEom} shows that
$$
R(t)E_\om^*R(t)=\cF(|\psi(t,\cdot)|^2)(\om)R(t)\,,
$$
and hence
$$
R(t)(E_\om^*-\cF(|\psi(t,\cdot)|^2)(\om)I_\fH)R(t)=0\,.
$$
Therefore
$$
\ba
\cM_N(t)(E_\om^*-\cF(|\psi(t,\cdot)|^2)(\om)I_\fH)=&\cM_N(t)(P(t)E_\om^*P(t)-\cF(|\psi(t,\cdot)|^2)(\om)P(t))
\\
&+\cM_N(t)(P(t)E_\om^*R(t)+R(t)E_\om^*P(t))\,,
\ea
$$
since $R(t)P(t)=P(t)R(t)=0$. Thus
$$
\ba
\cC(V,\cM_N(t)-\cR(t),\cM_N(t))(R(t))
\\
=\int_{\bR^3}\hat V(\om)((\cM_N(t)(P(t)E_\om^*P(t)-\cF(|\psi(t,\cdot)|^2)(\om)P(t))\cM_N(t)(P(t)E_\om R(t))
\\
-\cM_N(t)(R(t)E_\om P(t))(\cM_N(t)(P(t)E_\om^*P(t)-\cF(|\psi(t,\cdot)|^2)(\om)P(t)))\tfrac{d\om}{(2\pi)^3}
\\
+\int_{\bR^3}\hat V(\om)(\cM_N(t)(P(t)E_\om^*R(t)+R(t)E_\om^*P(t))\cM_N(t)(P(t)E_\om R(t))
\\
-\cM_N(t)(R(t)E_\om P(t))\cM_N(t)(P(t)E_\om^*R(t)+R(t)E_\om^*P(t)))\tfrac{d\om}{(2\pi)^3}
\\
+\tfrac1N\cM_N(t)[(V\star|\psi(t,\cdot)|^2),R(t)]&\,.
\ea
$$

Using again \eqref{Veven} implies that
$$
\ba
\int_{\bR^3}\hat V(\om)(\cM_N(t)(R(t)E_\om^*P(t))\cM_N(t)(P(t)E_\om R(t))d\om
\\
=\int_{\bR^3}\hat V(\om)(\cM_N(t)(R(t)E_\om P(t))\cM_N(t)(P(t)E_\om^*R(t))d\om&\,,
\ea
$$
so that
$$
\ba
\cC(V,\cM_N(t)-\cR(t),\cM_N(t))(R(t))
\\
=\int_{\bR^3}\hat V(\om)((\cM_N(t)(P(t)E_\om^*P(t)-\cF(|\psi(t,\cdot)|^2)(\om)P(t))\cM_N(t)(P(t)E_\om R(t))
\\
-\cM_N(t)(R(t)E_\om P(t))(\cM_N(t)(P(t)E_\om^*P(t)-\cF(|\psi(t,\cdot)|^2)(\om)P(t)))\tfrac{d\om}{(2\pi)^3}
\\
+\int_{\bR^3}\hat V(\om)(\cM_N(t)(P(t)E_\om^*R(t))\cM_N(t)(P(t)E_\om R(t))
\\
-\cM_N(t)(R(t)E_\om P(t))\cM_N(t)(R(t)E_\om^*P(t)))\tfrac{d\om}{(2\pi)^3}
\\
+\tfrac1N\cM_N(t)[(V\star|\psi(t,\cdot)|^2),R(t)]&\,.
\ea
$$
By \eqref{VR=}, one can further simplify the term
$$
\ba
\int_{\bR^3}\hat V(\om)\cF(|\psi(t,\cdot)|^2)(\om)&\cM_N(t)(P(t))\cM_N(t)(P(t)E_\om R(t))
\\
&-\cM_N(t)(R(t)E_\om P(t))(\cM_N(t)(P(t)))\tfrac{d\om}{(2\pi)^3}
\\
=&\cM_N(t)(P(t))\cM_N(t)(P(t)(V\star|\psi(t,\cdot)|^2)R(t))
\\
&-\cM_N(t)(R(t)(V\star|\psi(t,\cdot)|^2)P(t))\cM_N(t)(P(t))\,.
\ea
$$

Finally
$$
\cC(V,\cM_N(t)-\cR(t),\cM_N(t))(R(t))=T_1+T_2+T_3+T_4
$$
with
$$
\ba
T_1:=&\int_{\bR^3}\hat V(\om)((\cM_N(t)(P(t)E_\om^*P(t))\cM_N(t)(P(t)E_\om R(t))
\\
&-\cM_N(t)(R(t)E_\om P(t))(\cM_N(t)(P(t)E_\om^*P(t)))\tfrac{d\om}{(2\pi)^3}
\\
T_2:=&\cM_N(t)(R(t)(V\star|\psi(t,\cdot)|^2)P(t))\cM_N(t)(P(t))
\\
&-\cM_N(t)(P(t))\cM_N(t)(P(t)(V\star|\psi(t,\cdot)|^2)R(t))
\\
T_3:=&\int_{\bR^3}\hat V(\om)(\cM_N(t)(P(t)E_\om^*R(t))\cM_N(t)(P(t)E_\om R(t))
\\
&-\cM_N(t)(R(t)E_\om P(t))\cM_N(t)(R(t)E_\om^*P(t)))\tfrac{d\om}{(2\pi)^3}
\\
T_4:=&\tfrac1N\cM_N(t)[(V\star|\psi(t,\cdot)|^2),R(t)]\,.
\ea
$$
Observe again that all the integrals in the right hand side of the equalities defining $T_1$ and $T_3$ are of the form defined in Theorem \ref{T-DefC}, or of the form (I), (II)  or (III), or their adjoint, in the Remark following Theorem \ref{T-DefC}.

That
$$
\ba
T_1=&\cT(V,\cM_N(t)(P(t)\bu P(t)),\cM_N(t)(P(t)\bu R(t)))
\\
&-\cT(V,\cM_N(t)(R(t)\bu P(t)),\cM_N(t)(P(t)\bu P(t)))
\\
T_3=&\cT(V,\cM_N(t)(P(t)\bu R(t)),\cM_N(t)(P(t)\bu R(t)))
\\
&-\cT(V,\cM_N(t)(R(t)\bu P(t)),\cM_N(t)(R(t)\bu P(t)))
\ea
$$
follows from \eqref{Veven} and the definition \eqref{Def-cT}. This concludes the proof of Lemma \ref{L-Decomp-cC}.

%%%%%%%%%%%%%%%%%%%%%%%%%%%%%%%%%%%%%%%%%%%%%%%%%%%%%%%%%%%%%%%%%%%%%%%%%%%%%%%%%%%%%%%%%%%%%%%%%%%%%%%%%%%%%%%%%%%%%%%%%

\section{Proof of Lemma \ref{L-IneqT1234}}

%%%%%%%%%%%%%%%%%%%%%%%%%%%%%%%%%%%%%%%%%%%%%%%%%%%%%%%%%%%%%%%%%%%%%%%%%%%%%%%%%%%%%%%%%%%%%%%%%%%%%%%%%%%%%%%%%%%%%%%%%

In the sequel, we shall estimate these four terms in increasing order of technical difficulty.

\subsection{Bound for $T_4$}
%%%%%%%%%%%%%%%%%%%%%%%%%%%%%%%%%%%%%%%%%%%%%%%%%%%%%%%%%%%%%%%%%%%%%%%%%%%%%%%%%%%%%%%%%%%%%%%%%%%%%%%%%%%%%%%%%%%%%%%%%

The easiest term to treat is obviously $T_4$. We first recall that
\be\lb{MNContract}
\|\cM_N(t)(A)\|\le\|A\|\qquad\text{ for each }A\in\cL(\fH)
\ee
--- see the formula following (41) on p. 1041 in \cite{FGTPaul}. Thus
$$
\ba
\|T_4\|\le&\tfrac1N\|[V\star|\psi(t,\cdot)|^2,R(t)]\|\le\tfrac1N(\|R(t)V\star|\psi(t,\cdot)|^2\|+\|(V\star|\psi(t,\cdot)|^2)R(t)\|)
\\
=&\tfrac2N\|(V\star|\psi(t,\cdot)|^2)R(t)\|\,,
\ea
$$
where the equality follows from the fact that $R(t)=R(t)^*$, which implies that 
\be\lb{VR*=}
((V\star|\psi(t,\cdot)|^2)R(t))^*=R(t)(V\star|\psi(t,\cdot)|^2)\,.
\ee

On the other hand, by Jensen's inequality
$$
(|V|\star|\psi(t,\cdot)|^2)^2\le V^2\star|\psi(t,\cdot)|^2\,,
$$
so that
\be\lb{VRR?}
\ba
\|(V\star|\psi(t,\cdot)|^2)R(t)\|^2\le&\|V\star|\psi(t,\cdot)|^2\|^2_{L^\infty}
\\
\le&\|\,|V|\star|\psi(t,\cdot)|^2\,\|^2_{L^\infty}\le\|(V^2)\star|\psi(t,\cdot)|^2\|_{L^\infty}=\ell(t)^2\,,
\ea
\ee
and therefore
\be\lb{BdT4}
\|T_4\|\le\tfrac2N\ell(t)\,.
\ee

Finally, we recall that
$$
(\cM_N^{in}A)^*=\tfrac1N\sum_{k=1}^N(J_kA)^*=\tfrac1N\sum_{k=1}^NJ_k(A^*)=\cM_N^{in}(A^*)
$$
for each $A\in\cL(\fH)$, so that
\be\lb{MN*hom}
\ba
(\cM_N(t)A)^*=(\cU_N(t)^*(\cM_N^{in}A)\cU_N(t))^*=&\cU_N(t)^*(\cM_N^{in}A)^*\cU_N(t)
\\
=&\cU_N(t)^*\cM_N^{in}(A^*)\cU_N(t)=\cM_N(t)(A^*)\,.
\ea
\ee
Then \eqref{VR*=} and \eqref{MN*hom} imply that
$$
\ba
(\cM_N(t)[V\star|\psi(t,\cdot)|^2,R(t)])^*=&\cM_N(t)([V\star|\psi(t,\cdot)|^2,R(t)]^*)
\\
=&\cM_N(t)(-[V\star|\psi(t,\cdot)|^2,R(t)])
\\
=&-\cM_N(t)([V\star|\psi(t,\cdot)|^2,R(t)])
\ea
$$
so that $T_4^*=-T_4$. Hence $\pm iT_4$ are self-adjoint operators on $\fH_N$, so that
\be\lb{T4?}
\|T_4\|\le\tfrac2N\ell(t)\implies\pm iT_4\le\tfrac2N\ell(t)I_{\fH_N}\,.
\ee

\subsection{Bound for $T_2$}
%%%%%%%%%%%%%%%%%%%%%%%%%%%%%%%%%%%%%%%%%%%%%%%%%%%%%%%%%%%%%%%%%%%%%%%%%%%%%%%%%%%%%%%%%%%%%%%%%%%%%%%%%%%%%%%%%%%%%%%%%

Set
\be\lb{S2=}
S_2:=\cM_N(t)(P(t))\cM_N(t)(P(t)(V\star|\psi(t,\cdot)|^2)R(t))\,.
\ee
One has
$$
\ba
S_2=&\cU_N(t)^*\cM_N^{in}(P(t))\cM_N^{in}(P(t)(V\star|\psi(t,\cdot)|^2)R(t))\cU_N(t)
\\
=&\tfrac1N\sum_{k=1}^N\cU_N(t)^*(J_kP(t))\cM_N^{in}(P(t)(V\star|\psi(t,\cdot)|^2)R(t))(J_kP(t))\cU_N(t)
\\
&+\tfrac1N\sum_{k=1}^N\cU_N(t)^*(J_kP(t))[J_kP(t),\cM_N^{in}(P(t)(V\star|\psi(t,\cdot)|^2)R(t))]\cU_N(t)\,.
\ea
$$
Then
$$
\ba
{}[J_kP(t),\cM_N^{in}(P(t)(V\star|\psi(t,\cdot)|^2)R(t))]
\\
=\tfrac1N[J_kP(t),J_k(P(t)(V\star|\psi(t,\cdot)|^2)R(t))]
\\
=\tfrac1NJ_k(P(t)(V\star|\psi(t,\cdot)|^2)R(t))&\,,
\ea
$$
so that
$$
\ba
S_2=&\tfrac1N\sum_{k=1}^N\cU_N(t)^*(J_kP(t))\cM_N^{in}(P(t)(V\star|\psi(t,\cdot)|^2)R(t))(J_kP(t))\cU_N(t)
\\
&+\tfrac1{N^2}\sum_{k=1}^N\cU_N(t)^*J_k(P(t)(V\star|\psi(t,\cdot)|^2)R(t))\cU_N(t)\,.
\ea
$$

By cyclicity of the trace, for each $F_N^{in}$ satisfying \eqref{SymDens}, denoting 
$$
F_N(t):=\cU_N(t)F_N^{in}\cU_N(t)^*\,,
$$
one has
$$
\ba
\Tr_{\fH_N}(S_2F_N^{in})
\\
=\tfrac1N\sum_{k=1}^N\Tr_{\fH_N}(\cM_N^{in}(P(t)(V\star|\psi(t,\cdot)|^2)R(t))(J_kP(t))F_N(t)(J_kP(t)))
\\
+\tfrac1{N^2}\sum_{k=1}^N\Tr_{\fH_N}(J_k(P(t)(V\star|\psi(t,\cdot)|^2)R(t))F_N(t))
\ea
$$
so that
\be\lb{BdS2}
\ba
|\Tr_{\fH_N}(S_2F_N^{in})|
\\
\le\tfrac1N\sum_{k=1}^N\|\cM_N^{in}(P(t)(V\star|\psi(t,\cdot)|^2)R(t))\|\|(J_kP(t))F_N(t)(J_kP(t)))\|_1
\\
+\tfrac1{N^2}\sum_{k=1}^N\|J_k(P(t)(V\star|\psi(t,\cdot)|^2)R(t))\|\|F_N(t)\|_1
\\
\le\|(V\star|\psi(t,\cdot)|^2)R(t)\|\tfrac1N\sum_{k=1}^N\Tr_{\fH_N}((J_kP(t))F_N(t)(J_kP(t))))
\\
+\|(V\star|\psi(t,\cdot)|^2)R(t)\|\tfrac1{N^2}\sum_{k=1}^N\|F_N(t)\|_1
\\
=\|(V\star|\psi(t,\cdot)|^2)R(t)\|(\Tr_{\fH_N}(\cM_N(t)(P(t))F_N^{in})+\tfrac1N\|F_N^{in}\|_1)&\,.
\ea
\ee

By \eqref{MN*hom}, 
$$
S_2^*=\cM_N(t)(R(t)(V\star|\psi(t,\cdot)|^2)P(t))\cM_N(t)(P(t))\,,
$$
so that
$$
T_2=S_2^*-S_2=-T_2^*\,.
$$
Thus
$$
\ba
|\Tr_{\fH_N}(T_2F_N^{in})|\le&|\Tr_{\fH_N}(S_2^*F_N^{in})|+|\Tr_{\fH_N}(S_2F_N^{in})|
\\
=&|\Tr_{\fH_N}(F_N^{in}S_2)|+|\Tr_{\fH_N}(S_2F_N^{in})|=2|\Tr_{\fH_N}(S_2F_N^{in})|\,,
\ea
$$
so that
\be\lb{BdT2}
\ba
|\Tr_{\fH_N}(T_2F_N^{in})|
\\
\le2\|(V\star|\psi(t,\cdot)|^2)R(t)\|(\Tr_{\fH_N}(\cM_N(t)(P(t))F_N^{in})+\tfrac1N\|F_N^{in}\|_1)
\\
\le2\ell(t)(\Tr_{\fH_N}(\cM_N(t)(P(t))F_N^{in})+\tfrac1N\Tr_{\fH_N}(F_N^{in}))
\ea
\ee
by \eqref{VRR?}.

\smallskip
Next we use the following elementary observation.

\begin{Lem}\lb{L->0}
Let $T=T^*\in\cL(\fH_N)$ satisfy
$$
U_\si TU_\si^*=T\text{ for all }\si\in\fS_N\,,\quad\text{ and }\Tr_{\fH_N}(TF)\ge 0
$$
for each $F\in\cL(\fH_N)$ satisfying \eqref{SymDens}. Then $T\ge 0$. 
\end{Lem}

\begin{proof} Indeed, we seek to prove that 
$$
\la\Psi|T|\Psi\ra\ge 0\qquad\text{ for each }\Psi\in\fH_N\,.
$$
For each $\Psi\in\fH_N$ such that $\|\Psi_N\|_{\fH_N}=1$, set 
$$
F=\tfrac1{N!}\sum_{\si\in\fS_N}|U_\si\Psi\ra\la U_\si\Psi|\,.
$$
Then $F$ satisfies \eqref{SymDens}, so that
$$
0\le\Tr(TF)=\tfrac1{N!}\sum_{\si\in\fS_N}\la U_\si\Psi|T|U_\si\Psi\ra=\tfrac1{N!}\sum_{\si\in\fS_N}\la\Psi|U_\si^*TU_\si|\Psi\ra=\la\Psi|T|\Psi\ra
$$
since $U_\si^*TU_\si=T$ for each $\si\in\fS_N$. Thus $\la\Psi|T|\Psi\ra\ge 0$ for each $\Psi\in\fH_N$ such that $\|\Psi_N\|_{\fH_N}=1$, and thus for each $\Psi\in\fH_N\setminus\{0\}$ by normalization.
\end{proof}

\smallskip
The inequality \eqref{BdT2} implies that
$$
\ba
2\ell(t)\Tr_{\fH_N}((\cM_N(t)(P(t))+\tfrac1NI_{\fH_N})F_N^{in})\ge&|\Tr_{\fH_N}(T_2F_N^{in})|
\\
\ge&\Tr_{\fH_N}(\pm iT_2F_N^{in})\,,
\ea
$$
and we conclude from Lemma \ref{L->0} that
\be\lb{T_2?}
\pm iT_2\le 2\ell(t)(\cM_N(t)(P(t))+\tfrac1NI_{\fH_N})\,.
\ee

\subsection{Bound for $T_1$}
%%%%%%%%%%%%%%%%%%%%%%%%%%%%%%%%%%%%%%%%%%%%%%%%%%%%%%%%%%%%%%%%%%%%%%%%%%%%%%%%%%%%%%%%%%%%%%%%%%%%%%%%%%%%%%%%%%%%%%%%%

Next we estimate 
$$
\ba
S_1:=&\cT(V,\cM_N(t)(P(t)\bu P(t)),\cM_N(t)(P(t)\bu R(t)))
\\
=&\cU_N(t)^*\cT(V,\cM_N^{in}(P(t)\bu P(t)),\cM_N^{in}(P(t)\bu R(t)))\cU_N(t)
\\
=&\tfrac1N\sum_{k=1}^N\cU_N(t)^*\cT(V,J_k(P(t)\bu P(t)),\cM_N^{in}(P(t)\bu R(t)))\cU_N(t)\,.
\ea
$$
Observe that
$$
\ba
\cT(V,J_k(P(t)\bu P(t)),\cM_N^{in}(P(t)\bu R(t)))
\\
=\int_{\bR^3}\hat V(\om)(J_kP(t))J_k(P(t)E_\om^*P(t))\cM_N^{in}(P(t)E_\om R(t))(J_kP(t))\tfrac{d\om}{(2\pi)^3}
\\
+\tfrac1N\int_{\bR^3}\hat V(\om)(J_kP(t))J_k(P(t)E_\om^*P(t))[J_kP(t),J_k(P(t)E_\om R(t))]\tfrac{d\om}{(2\pi)^3}&\,,
\ea
$$
since $P(t)=P(t)^2$, so that $J_kP(t)=(J_kP(t))^2$. Then
$$
[J_kP(t),J_k(P(t)E_\om R(t))]=J_k(P(t)E_\om R(t))\,,
$$
so that
$$
\ba
J_k(P(t)E_\om^*P(t))[J_kP(t),J_k(P(t)E_\om R(t))]=J_k(P(t)E_\om^*P(t)E_\om R(t))
\\
=J_k(P(t)E_\om^*(I-R(t))E_\om R(t))=-\cF(|\psi(t,\cdot)|^2)(-\om)J_k(P(t)E_\om^*R(t))&\,.
\ea
$$
Hence \eqref{Veven} implies that
$$
\ba
\tfrac1N\int_{\bR^3}\hat V(\om)(J_kP(t))J_k(P(t)E_\om^*P(t))[J_kP(t),J_k(P(t)E_\om R(t))]\tfrac{d\om}{(2\pi)^3}
\\
=-\tfrac1N(J_kP(t))J_k(P(t)(V\star|\psi(t,\cdot)|^2)R(t))&\,.
\ea
$$

On the other hand
$$
\ba
\int_{\bR^3}\hat V(\om)(J_kP(t))J_k(P(t)E_\om^*P(t))\cM_N^{in}(P(t)E_\om R(t))(J_kP(t))\tfrac{d\om}{(2\pi)^3}
\\
=\tfrac1N\int_{\bR^3}\hat V(\om)(J_kP(t))J_k(P(t)E_\om^*P(t))\sum_{l=1\atop l\not=k}^NJ_l(P(t)E_\om R(t))J_kP(t)\tfrac{d\om}{(2\pi)^3}
\\
=\tfrac1N\int_{\bR^3}\hat V(\om)(J_kP(t))J_k(E_\om^*)\sum_{l=1\atop l\not=k}^NJ_l(P(t)E_\om R(t))J_kP(t)\tfrac{d\om}{(2\pi)^3}
\\
=(J_kP(t))\left(\tfrac1N\sum_{l=1\atop l\not=k}^N(J_kP(t))(J_lP(t))V_{kl}(J_lR(t))(J_kP(t))\right)(J_kP(t))&\,,
\ea
$$
since $J_k(P(t)E_\om R(t))J_k(P(t))=0$, with $V_{kl}$ defined as in \eqref{DefVkl}.

Hence
$$
\ba
S_1=&\tfrac1{N^2}\sum_{1\le k\not=l\le N}\cU_N(t)^*(J_kP(t))^2(J_lP(t))V_{kl}(J_lR(t))(J_kP(t))^2\cU_N(t)
\\
&-\tfrac1{N^2}\sum_{k=1}^N\cU_N(t)^*(J_kP(t))J_k(P(t)(V\star|\psi(t,\cdot)|^2)R(t))\cU_N(t)\,.
\ea
$$
Therefore, by cyclicity of the trace, for each $F_N^{in}\in\cL(\fH_N)$ satisfying \eqref{SymDens}, denoting $F_N(t):=\cU_N(t)F_N^{in}\cU_N(t)^*$, one has
$$
\ba
\Tr_{\fH_N}(S_1F_N^{in})
\\
=\tfrac1{N^2}\sum_{1\le k\not=l\le N}\Tr_{\fH_N}((J_kP(t))(J_lP(t))V_{kl}(J_lR(t))(J_kP(t))^2F_N(t)(J_kP(t)))
\\
-\tfrac1{N^2}\sum_{k=1}^N\Tr_{\fH_N}(J_k(P(t)(V\star|\psi(t,\cdot)|^2)R(t))F_N(t)(J_kP(t)))&\,,
\ea
$$
so that
\be\lb{BdS1}
\ba
|\Tr_{\fH_N}(S_1F_N^{in})|
\\
\le\!\tfrac1{N^2}\!\sum_{1\le k\not=l\le N}\!\|\!(J_kP(t))\!(J_lP(t))\!V_{kl}\!(J_lR(t))\!(J_kP(t))\!\|\|\!(J_kP(t))\!F_N(t)\!(J_kP(t)))\!\|_1
\\
+\tfrac1{N^2}\sum_{k=1}^N\|J_k(P(t)(V\star|\psi(t,\cdot)|^2)R(t))\|\|F_N(t)\|_1\|J_kP(t))\|
\\
\le(1-\tfrac1N)\|V_{12}J_2R(t)\|\Tr_{\fH_N}(F_N^{in}\cM_N(t)(P(t)))
\\
+\tfrac2N\|(V\star|\psi(t,\cdot)|^2)R(t)\|\|F_N(t)\|_1&\,.
\ea
\ee

Finally
$$
\ba
T_1=&\cT(V,\cM_N(t)(P(t)\bu P(t)),\cM_N(t)(P(t)\bu R(t)))
\\
&-\cT(V,\cM_N(t)(R(t)\bu P(t)),\cM_N(t)(P(t)\bu P(t)))=S_1-S_1^*=-T_1^*
\ea
$$
because of Lemma \ref{L-SelfAd}, so that
$$
\ba
|\Tr_{\fH_N}(T_1F_N^{in})|\le&2(1-\tfrac1N)\|V_{12}J_2R(t)\|\Tr_{\fH_N}(F_N^{in}\cM_N(t)(P(t)))
\\
&+\tfrac4N\|(V\star|\psi(t,\cdot)|^2)R(t)\|\|F_N(t)\|_1\,.
\ea
$$
Since $R(s)$ is a rank-one orthogonal projection
\be\lb{V12J2R?}
\ba
\|V_{12}(J_2R(t))\|^2=&\|(J_2R(t))V^2_{12}(J_2R(t))\|
\\
=&\|(V^2\star|\psi(t,\cdot)|^2)\otimes R(s)\|\le\|(V^2\star|\psi(t,\cdot)|^2\|_{L^\infty}=\ell(t)^2\,.
\ea
\ee
Thus
\be\lb{BdT1}
\ba
|\Tr_{\fH_N}(T_1F_N^{in})|\le&2\ell(t)\left((1\!-\!\tfrac1N)\Tr_{\fH_N}(F_N^{in}\cM_N(t)(P(t)))\!+\!\tfrac2N\|F_N^{in}\|_1\right)
\\
=&2\ell(t)\left((1\!-\!\tfrac1N)\Tr_{\fH_N}(F_N^{in}\cM_N(t)(P(t)))\!+\!\tfrac2N\Tr_{\fH_N}(F_N^{in})\right).
\ea
\ee

In particular
$$
2\ell(t)\Tr_{\fH_N}\left(F_N^{in}\left((1\!-\!\tfrac1N)\cM_N(t)(P(t)))+\tfrac2NI_{\fH_N}\right)\right)\ge\Tr_{\fH_N}(\pm iT_1F_N^{in})
$$
and since this inequality holds for each $F_N^{in}\in\cL_s(\fH_N)$ such that $F_N^{in}=(F_N^{in})^*\ge 0$, we conclude from Lemma \ref{L->0} that
\be\lb{T1?}
\pm iT_1\le 2\ell(t)\left((1\!-\!\tfrac1N)\cM_N(t)(P(t)))+\tfrac2NI_{\fH_N}\right)
\ee

\subsection{The Operator $\Pi_N$}
%%%%%%%%%%%%%%%%%%%%%%%%%%%%%%%%%%%%%%%%%%%%%%%%%%%%%%%%%%%%%%%%%%%%%%%%%%%%%%%%%%%%%%%%%%%%%%%%%%%%%%%%%%%%%%%%%%%%%%%%%

In order to treat the last term $T_3$, we need the following auxiliary lemma --- see the formula preceding (13) in \cite{Pickl1}.

\begin{Lem}\lb{L-LemPi}
Let $R=R^*$ be a rank-one projection on $\fH$ and let $P:=I-R$. Set $\Pi_N:=\cM_N^{in}P$. For each $N>1$, 
$$
\Pi^*_N=\Pi_N\,,\quad\Pi_N^2\ge\tfrac1N\Pi_N\,,\qquad\text{ and }\quad\Ker\Pi_N=\Ker(I-R^{\otimes N})\,,
$$
so that
$$
\Pi_N\ge\tfrac1N(1-R^{\otimes N})\,.
$$
In particular, there exists a pseudo-inverse $\Pi_N^{-1}:\,(\Ker\Pi_N)^\perp\to(\Ker\Pi_N)^\perp$, with extension by $0$ on $\Ker\Pi_N$ also (abusively) denoted $\Pi_N$, such that
\be\lb{PiPi-1}
\Pi_N^{-1}\Pi_N=\Pi_N\Pi_N^{-1}=I-R^{\otimes N}\,.
\ee
\end{Lem}

In \cite{Pickl1}, the definition of the pseudo-inverse of $\Pi_N$ immediately follows from formula (6), which can be viewed as the spectral decomposition of $\Pi_N$. The proof below is quite straightforward and avoids using the clever argument 
leading to formula (6) in \cite{Pickl1}, which is not entirely obvious unless one already knows the result.

\begin{proof}
That $\Pi_N$ is self-adjoint is obvious by definition of $\cM_N^{in}$. Then
$$
\Pi_N^2=\frac1{N^2}\left(\sum_{k=1}^NJ_kP+2\sum_{1\le k<l\le N}J_kPJ_lP\right)\ge\frac1{N^2}\sum_{k=1}^NJ_kP=\frac1N\Pi_N\,.
$$
If $X\in\Ker\Pi_N$, one has, for each $k=1,\ldots,N$,
$$
0=\sum_{k=1}^N\la X|J_kP|X\ra\implies\la X|J_kP|X\ra=0\implies J_kPX=0\,.
$$
Hence
$$
X=J_NRX=J_{N-1}RX=\ldots=J_2RX=J_1RX
$$
so that
$$
X=J_NRX=J_NRJ_{N-1}RX=\ldots=J_NRJ_{N-1}R\ldots J_2RJ_1RX=R^{\otimes N}X\,.
$$
Thus $\Ker\Pi_N=\Ker(I-R^{\otimes N})$. Finally
$$
\Pi_N^3=\Pi_N^{1/2}\Pi_N^2\Pi_N^{1/2}\ge\frac1N\Pi_N^{1/2}\Pi_N\Pi_N^{1/2}=\frac1N\Pi_N^2\,.
$$
Therefore, for each $X\in\fH_N$, one has
$$
\la\Pi_NX|\Pi_N|\Pi_NX\ra\ge\frac1N\|\Pi_NX\|^2\,.
$$
Since $\Pi_N=\Pi_N^*$, one has
$$
\overline{\IM\Pi_N}=(\Ker\Pi_N)^\perp
$$
(see for instance Corollary 2.18 (iv) in chapter 2 of \cite{Brezis}). Since
$$
\la Y|\Pi_N|Y\ra\ge\frac1N\|Y\|^2\,,\qquad Y\in\IM\Pi_N\,,
$$
and since one has obviously $\|\Pi_N\|\le 1$, a straightforward density argument shows that
$$
\la Y|\Pi_N|Y\ra\ge\frac1N\|Y\|^2\,,\qquad Y\in(\Ker\Pi_N)^\perp\,.
$$
Hence
$$
\Pi_N\ge\tfrac1N(1-R^{\otimes N})\,.
$$
The existence of the pseudo-inverse $\Pi_N^{-1}$ follows from this inequality.
\end{proof}

\subsection{Bound for $T_3$}
%%%%%%%%%%%%%%%%%%%%%%%%%%%%%%%%%%%%%%%%%%%%%%%%%%%%%%%%%%%%%%%%%%%%%%%%%%%%%%%%%%%%%%%%%%%%%%%%%%%%%%%%%%%%%%%%%%%%%%%%%

Finally, we treat the term $T_3$. Set
$$
\ba
S_3=&\cT(V,\cM_N(t)(P(t)\bu R(t)),\cM_N(t)(P(t)\bu R(t)))
\\
=&\cU_N(t)^*\cT(V,\cM_N^{in}(P(t)\bu R(t)),\cM_N^{in}(P(t)\bu R(t)))\cU_N(t)\,.
\ea
$$
One easily checks that
$$
\ba
\cT(V,\cM_N^{in}(P(t)\bu R(t)),\cM_N^{in}(P(t)\bu R(t)))
\\
=\int_{\bR^3}\hat V(\om)\cM_N^{in}(P(t)E_\om^*R(t))\cM_N^{in}(P(t)E_\om R(t))\tfrac{d\om}{(2\pi)^3}
\\
=\tfrac1{N^2}\sum_{1\le k\not=l\le N}(J_lP(t))(J_kP(t))V_{kl}(J_kR(t))(J_lR(t))&\,.
\ea
$$

At this point, we set $\Pi_N(t):=\cM_N^{in}P(t)$ and use Lemma \ref{L-LemPi} to define the pseudo-inverse $\Pi_N(t)^{-1}$. One has $\Pi_N(t)=\Pi_N(t)^*\ge 0$, so that $\Pi_N(t)^{-1}=(\Pi_N(t)^{-1})^*\ge 0$ on $\Ker(I-R(t)^{\otimes N})$.
Abusing the notation $\Pi_N(t)^{-1/2}$ to designate the linear map $(\Pi_N(t)^{-1})^{1/2}$, we deduce from \eqref{PiPi-1} that
$$
\Pi_N(t)^{1/2}\Pi_N(t)^{-1/2}=I-R(t)^{\otimes N}\,,
$$
so that 
$$
(J_kP(t))\Pi_N(t)^{1/2}\Pi_N(t)^{-1/2}=\Pi_N(t)^{1/2}\Pi_N(t)^{-1/2}(J_kP(t))=J_lP(t)\,.
$$

Hence
$$
\ba
\cT(V,\cM_N^{in}(P(t)\bu R(t)),\cM_N^{in}(P(t)\bu R(t)))
\\
=\tfrac1{N^2}\!\sum_{1\le k\not=l\le N}\!(J_lP(t))(J_kP(t))\Pi_N(t)^{-\frac12}\Pi_N(t)^{\frac12}V_{kl}(J_kR(t))(J_lR(t))&\,,
\ea
$$
and we study the quantity
$$
\ba
\Tr_{\fH_N}(S_3F_N^{in})=\Tr_{\fH_N}((F_N^{in})^\frac12S_2(F_N^{in})^\frac12)
\\
=\Tr_{\fH_N}(F_N(t)^\frac12\cT(V,\cM_N^{in}(P(t)\bu R(t)),\cM_N^{in}(P(t)\bu R(t)))F_N(t)^\frac12)&\,.
\ea
$$
where $F_N(t)=\cU_N(t)F_N^{in}\cU_N(t)^*$, for each $F_N^{in}\in\cL(\fH_N)$ satisfying \eqref{SymDens}. Observe that
$$
\ba
|\!\Tr_{\fH_N}\!(F_N(t)^\frac12\!(J_lP(t))(J_kP(t))\Pi_N(t)^{-\frac12}\!\Pi_N(t)^{\frac12}\!V_{kl}(J_kR(t))(J_lR(t))F_N(t)^\frac12\!)|
\\
\le\|\Pi_N(t)^{-\frac12}(J_kP(t))(J_lP(t))F_N(t)^\frac12\|_2\|\Pi_N(t)^{\frac12}V_{kl}(J_kR(t))(J_lR(t))F_N(t)^\frac12\|_2&,
\ea
$$
so that, by the Cauchy-Schwarz inequality,
$$
\ba
|\Tr_{\fH_N}(S_3F_N^{in})|\le
&\tfrac1{N^2}\left(\sum_{1\le k\not=l\le N}\|\Pi_N(t)^{-\frac12}(J_kP(t))(J_lP(t))F_N(t)^\frac12\|_2^2\right)^{1/2}
\\
&\times\left(\sum_{1\le k\not=l\le N}\|\Pi_N(t)^{\frac12}V_{kl}(J_kR(t))(J_lR(t))F_N(t)^\frac12\|_2^2\right)^{1/2}\,.
\ea
$$

First, one has
$$
\ba
\|\Pi_N(t)^{-\frac12}(J_kP(t))(J_lP(t))F_N(t)^\frac12\|_2^2
\\
=\Tr_{\fH_N}(F_N(t)^\frac12(J_lP(t))(J_kP(t))\Pi_N(t)^{-1}(J_kP(t))(J_lP(t))F_N(t)^\frac12)
\\
=\Tr_{\fH_N}(F_N(t)^\frac12\Pi_N(t)^{-1}(J_kP(t))(J_lP(t))F_N(t)^\frac12)
\\
=\Tr_{\fH_N}(\Pi_N(t)^{-1}(J_kP(t))(J_lP(t))F_N(t))&\,,
\ea
$$
(the second equality follows from the fact that $J_k(P(t))$ commutes with $\Pi_N(t)$ and $\Pi_N(t)^{-1}$), so that
$$
\ba
\sum_{1\le k\not=l\le N}\|\Pi_N(t)^{-\frac12}(J_kP(t))(J_lP(t))F_N(t)^\frac12\|_2^2
\\
\le\Tr_{\fH_N}\left(\Pi_N(t)^{-1}\sum_{1\le k,l\le N}(J_kP(t))(J_lP(t))F_N(t)\right)
\\
=N^2\Tr_{\fH_N}(\Pi_N(t)^{-1}\Pi_N(t)^2F_N(t))=N^2\Tr_{\fH_N}(\Pi_N(t)F_N(t))&\,.
\ea
$$
The inequality above follows from the fact that
$$
\ba
\Tr_{\fH_N}(\Pi_N(t)^{-1}(J_kP(t))^2F_N(t))&
\\
=\Tr_{\fH_N}(F_N(t)^\frac12(J_kP(t))\Pi_N(t)^{-1}(J_kP(t))F_N(t)^\frac12)&\ge 0\,.
\ea
$$

On the other hand
$$
\ba
\sum_{1\le k\not=l\le N}\|\Pi_N(t)^{\frac12}V_{kl}(J_kR(t))(J_lR(t))F_N(t)^\frac12\|_2^2
\\
=\sum_{1\le k\not=l\le N}\Tr(F_N(t)^\frac12(J_lR(t))(J_kR(t))V_{kl}\Pi_N(t)V_{kl}(J_kR(t))(J_lR(t))F_N(t)^\frac12)
\\
=\tfrac1N\sum_{1\le k\not=l\le N}\|J_k(P(t))V_{kl}(J_kR(t))(J_lR(t))F_N(t)^\frac12\|_2^2
\\
+\tfrac1N\sum_{1\le k\not=l\le N}\|J_l(P(t))V_{kl}(J_kR(t))(J_lR(t))F_N(t)^\frac12\|_2^2
\\
+\tfrac1N\sum_{1\le m\not=k\not=l\le N}\|(J_mP(t))V_{kl}(J_kR(t))(J_lR(t))F_N(t)^\frac12\|_2^2
\\
\le\tfrac2N\sum_{1\le k\not=l\le N}\|V_{kl}(J_kR(t))(J_lR(t))F_N(t)^\frac12\|_2^2
\\
+\tfrac1N\sum_{1\le m\not=k\not=l\le N}\|(J_mP(t))V_{kl}(J_kR(t))(J_lR(t))F_N(t)^\frac12\|_2^2&\,.
\ea
$$
Now, $m\notin\{k,l\}$ implies that $J_mP(t)$ commutes with $V_{kl}$, $J_kR(t)$ and $J_lR(t)$, so that
$$
\ba
\|(J_mP(t))V_{kl}(J_kR(t))(J_lR(t))F_N(t)^\frac12\|_2^2
\\
=\|V_{kl}(J_kR(t))(J_lR(t))(J_mP(t))F_N(t)^\frac12\|_2^2
\\
\le\|V_{kl}(J_kR(t))(J_lR(t))\|^2\|(J_mP(t))F_N(t)^\frac12\|_2^2
\\
=\|V_{12}R(t)\otimes R(t)\|^2\Tr_{\fH_N}((J_mP(t))F_N(t)(J_mP(t)))
\\
=\|V_{12}R(t)\otimes R(t)\|^2\Tr_{\fH_N}(\Pi_N(t)F_N(t))&\,.
\ea
$$
Therefore
\be\lb{BdS3}
\ba
|\Tr_{\fH_N}(S_3F_N^{in})|
\\
\le\tfrac1{N^2}\cdot N\Tr_{\fH_N}(\Pi_N(t)F_N(t))^\frac12\left(\tfrac{2N(N-1)}N\|V_{12}R(t)\!\otimes\!R(t)\|^2\|F_N(t)\|_1\right.
\\
\left.\!+\!\tfrac{N(N-1)(N-2)}N\|V_{12}R(t)\!\otimes\!R(t)\|^2\Tr_{\fH_N}(\Pi_N(t)F_N(t))\right)^\frac12
\\
\le\!\tfrac{1}{\sqrt{N}}\|V_{12}R(t)\!\otimes\!R(t)\|\Tr_{\fH_N}(\cM_N(t)(P(t))F_N^{in})^\frac12
\\
\times\left(2\|F_N(t)\|_1+\!(N\!-\!2)\!\Tr_{\fH_N}(\cM_N(t)(P(t))F_N^{in})\right)^\frac12&.
\ea
\ee

Now
$$
\ba
T_3=&\cT(V,\cM_N(t)(P(t)\bu R(t)),\cM_N(t)(P(t)\bu R(t)))
\\
-&\cT(V,\cM_N(t)(R(t)\bu P(t)),\cM_N(t)(R(t)\bu P(t)))=S_3-S_3^*=-T_3^*\,,
\ea
$$
according to Lemma \ref{L-SelfAd}. Thus \eqref{BdS3} implies that
$$
\ba
|\Tr_{\fH_N}(T_3F_N^{in})|\le&2\|V_{12}R(t)\!\otimes\!R(t)\|\Tr_{\fH_N}(\cM_N(t)(P(t))F_N^{in})^\frac12
\\
&\times\left(\tfrac2N\|F_N(t)\|_1+\tfrac{N-2}N\Tr_{\fH_N}(\cM_N(t)(P(t))F_N^{in})\right)^\frac12\,.
\ea
$$
According to \eqref{V12J2R?}
$$
\|V_{12}R(t)\otimes R(t)\|=\|V_{12}(J_1R(t))(J_2R(t))\|\le\|V_{12}J_2R(t)\|\le\ell(t)\,,
$$
so that
\be\lb{BdT3}
|\Tr_{\fH_N}(T_3F_N^{in})|\le 2\ell(t)\left((1-\tfrac1N)\Tr_{\fH_N}(\cM_N(t)(P(t))F_N^{in})+\tfrac2N\|F_N^{in}\|_1\right)\,.
\ee
In particular
$$
\Tr_{\fH_N}(\pm iT_3F_N^{in})|\le 2\ell(t)\Tr_{\fH_N}\left(F_N^{in}\left((1-\tfrac1N)\cM_N(t)(P(t))+\tfrac1NI_{\fH_N}\right)\right)\,.
$$
Since this last inequality holds for each $F_N^{in}\in\cL(\fH_N)$ satisfying \eqref{SymDens}, we deduce from Lemma \ref{L->0} that
\be\lb{T3?}
\pm iT_3\le  2\ell(t)\left((1-\tfrac1N)\cM_N(t)(P(t))+\tfrac1NI_{\fH_N}\right)\,.
\ee

%%%%%%%%%%%%%%%%%%%%%%%%%%%%%%%%%%%%%%%%%%%%%%%%%%%%%%%%%%%%%%%%%%%%%%%%%%%%%%%%%%%%%%%%%%%%%%%%%%%%%%%%%%%%%%%%%%%%%%%%%

\section{Proofs of part (2) in Theorem \ref{T-OpIneq} and Corollary \ref{C-MF}}

%%%%%%%%%%%%%%%%%%%%%%%%%%%%%%%%%%%%%%%%%%%%%%%%%%%%%%%%%%%%%%%%%%%%%%%%%%%%%%%%%%%%%%%%%%%%%%%%%%%%%%%%%%%%%%%%%%%%%%%%%

\subsection{Proof of part (2) in Theorem \ref{T-OpIneq}}
%%%%%%%%%%%%%%%%%%%%%%%%%%%%%%%%%%%%%%%%%%%%%%%%%%%%%%%%%%%%%%%%%%%%%%%%%%%%%%%%%%%%%%%%%%%%%%%%%%%%%%%%%%%%%%%%%%%%%%%%%

Applying Lemma \ref{L-Decomp-cC} shows that 
$$
\pm i\cC(V,\cM_N(t)-\cR(t),\cM_N(t))(R(t))=\pm i(T_1+T_2+T_3+T_4)\,.
$$
With Lemma \ref{L-IneqT1234}, this shows that
$$
(\pm i\cC(V,\cM_N(t)-\cR(t),\cM_N(t))(R(t)))^*=\pm i\cC(V,\cM_N(t)-\cR(t),\cM_N(t))(R(t))
$$
and that
$$
\pm i\cC(V,\cM_N(t)-\cR(t),\cM_N(t))(R(t))\le 6\ell(t)\left(\cM_N(t)(P(t))+\tfrac2NI_{\fH_N}\right)\,.
$$

It remains to bound the function 
$$
\ell(t):=\|V^2\star|\psi(t,\cdot)|^2\|_{L^\infty(\bR^3)}^{1/2}\,.
$$
Since
$$
V=V_1+V_2\text{ with }V_1\in\cF L^1(\bR^3)\subset L^\infty(\bR^3)\text{ and }V_2\in L^2(\bR^3)
$$
one has
$$
\ba
0\le&V^2\star|\psi(t,\cdot)|^2\le 2V_1^2\star|\psi(t,\cdot)|^2+2V_2^2\star|\psi(t,\cdot)|^2
\\
\le&2\|V_1\|_{L^\infty(\bR^3)}^2\|\psi(t,\cdot)\|_{L^2(\bR^3)}^2+2\|V_2\|_{L^2(\bR^3)}^2\|\psi(t,\cdot)\|_{L^\infty(\bR^3)}^2\,.
\ea
$$
Minimizing $\|V_1\|_{L^\infty(\bR^3)}+\|V_2\|_{L^2(\bR^3)}$ over all possible decompositions of $V=V_1+V_2$ as above, one has
$$
\ba
0\le V^2\star|\psi(t,\cdot)|^2\le&4\|V\|^2_{L^2(\bR^3)+L^\infty(\bR^3)}\max(\|\psi(t,\cdot)\|_{L^2(\bR^3)}^2,\|\psi(t,\cdot)\|_{L^\infty(\bR^3)}^2)
\\
\le&4\|V\|^2_{L^2(\bR^3)+L^\infty(\bR^3)}\max(\|\psi(t,\cdot)\|_{L^2(\bR^3)}^2,C_S^2\|\psi(t,\cdot)\|_{H^2(\bR^3)}^2)
\\
\le&4\max(1,C_S)^2\|V\|^2_{L^2(\bR^3)+L^\infty(\bR^3)}\|\psi(t,\cdot)\|_{H^2(\bR^3)}^2=:L(t)^2\,.
\ea
$$

%%%%%%%%%%%%%%%%%%%%%%%%%%%%%%%%%%%%%%%%%%%%%%%%%%%%%%%%%%%%%%%%%%%%%%%%%%%%%%%%%%%%%%%%%%%%%%%%%%%%%%%%%%%%%%%%%%%%%%%%%

\subsection{Proof of Corollary \ref{C-MF}}

%%%%%%%%%%%%%%%%%%%%%%%%%%%%%%%%%%%%%%%%%%%%%%%%%%%%%%%%%%%%%%%%%%%%%%%%%%%%%%%%%%%%%%%%%%%%%%%%%%%%%%%%%%%%%%%%%%%%%%%%%

Pickl's functional defined in \cite{Pickl1} and recalled in formula \eqref{PicklQKlim} can be recast as
\be\lb{PicklFuncBis}
\a_N(t):=\Tr_{\fH}(F_{N:1}(t)P(t))
\ee
(see Definition 2.2 and formula (6) in \cite{Pickl1}), where $F_{N:1}(t)$ is the single-body reduced density operator deduced from
$$
F_N(t):=\cU_N(t)F_N^{in}\cU_N(t)^*\,,
$$
where $F_N^{in}\in\cL(\fH_N)$ satisfies \eqref{SymDens}. Specifically $F_{N:1}(t)$ is defined by the formula
$$
\Tr_\fH(F_{N:1}(t)A)=\Tr_{\fH_N}(F_N(t)J_1A)\,,\qquad\text{ for all }A\in\cL(\fH)\,.
$$
Since $F_N(t)$ satisfies \eqref{SymDens}, it holds
\be\lb{PteCharMN}
\ba
\Tr_\fH(F_{N:1}(t)A)=&\Tr_{\fH_N}(F_N(t)\cM_N^{in}A)
\\
=&\Tr_{\fH_N}(F_N^{in}\cM_N(t)A)\,,\qquad\text{ for all }A\in\cL(\fH)\,.
\ea
\ee
This is Lemma 2.3 in \cite{FGTPaul}, and the \textit{raison d'\^etre} of $\cM_N(t)$. Thus, formula \eqref{PicklQKlim} and \eqref{PicklFuncBis} are indeed equivalent.

\subsubsection{The Gronwall inequality for Pickl's functional}

%%%%%%%%%%%%%%%%%%%%%%%%%%%%%%%%%%%%%%%%%%%%%%%%%%%%%%%%%%%%%%%%%%%%%%%%%%%%%%%%%%%%%%%%%%%%%%%%%%%%%%%%%%%%%%%%%%%%%%%%%

One deduces from part (2) in Theorem \ref{T-OpIneq} that
$$
\ba
\cM_N(t)(P(t))=&\cM_N^{in}(P(0))+\tfrac1\hb\int_0^t-i\cC(V,\cM_N(s)-\cR(s),\cM_N(s))(R(s))ds
\\
\le&\cM_N^{in}(P(0))+\tfrac{6}\hb\int_0^tL(s)\left(\cM_N(s)(P(s))+\tfrac{2}{N}I_{\fH_N}\right)ds\,.
\ea
$$
This inequality implies that
$$
\ba
\Tr_{\fH_N}((F_N^{in})^\frac12\cM_N(t)(P(t))(F_N^{in})^\frac12)\le\Tr_{\fH_N}((F_N^{in})^\frac12\cM_N^{in}(P(0))(F_N^{in})^\frac12)
\\
+\tfrac{6}\hb\int_0^tL(s)\left(\Tr_{\fH_N}((F_N^{in})^\frac12\cM_N(s)(P(s))(F_N^{in})^\frac12)+\tfrac{2}{N}\Tr_{\fH_N}(F_N^{in})\right)ds&\,.
\ea
$$
Now, by cyclicity of the trace and \eqref{PteCharMN},
$$
\ba
\Tr_{\fH_N}((F_N^{in})^\frac12\cM_N(t)(P(t))(F_N^{in})^\frac12)=&\Tr_{\fH_N}(F_N^{in}\cM_N(t)(P(t)))
\\
=&\Tr_\fH(F_{N:1}(t)P(t))=\a_N(t)\,,
\ea
$$
so that, by Gronwall's inequality,
$$
\a_N(t)\le\a_N(0)\exp\left(\tfrac{6}{\hb}\int_0^tL(s)ds\right)+\tfrac{2}{N}\left(\exp\left(\tfrac{6}{\hb}\int_0^tL(s)ds\right)-1\right)\,.
$$
For instance, if $F_N^{in}=|\psi^{in}\ra\la\psi^{in}|^{\otimes N}$ with $\psi^{in}\in\fH$ and $\|\psi^{in}\|_\fH=1$, one has
$$
\a_N(0)=\Tr_{\fH_N}(R(0)^{\otimes N}\cM_N^{in}(P(0)))=\Tr_\fH(R(0)P(0))=0\,,
$$
so that
$$
\a_N(t)\le\frac{2}{N}\left(\exp\left(\tfrac{6}{\hb}\int_0^tL(s)ds\right)-1\right)=O\left(\frac1N\right)\,.
$$

\subsubsection{Pickl's functional and the trace norm}

%%%%%%%%%%%%%%%%%%%%%%%%%%%%%%%%%%%%%%%%%%%%%%%%%%%%%%%%%%%%%%%%%%%%%%%%%%%%%%%%%%%%%%%%%%%%%%%%%%%%%%%%%%%%%%%%%%%%%%%%%

How the inequality above implies the mean-field limit is explained by the following lemma, which recaps the results stated as Lemmas 2.1 and 2.2 in \cite{KPickl}, and whose proof is given below for the sake of keeping the present paper 
self-contained. 

If $F_N^{in}\in\cL(\fH_N)$ satisfies \eqref{SymDens}, for each $m=1,\ldots,N$, we denote by $F_{N:m}(t)$ the $m$-particle reduced density operator deduced from $F_N(t)=\cU_N(t)F_N^{in}\cU_N(t)^*$, i.e.
$$
\Tr_{\fH_m}(F_{N:m}(t)A_1\otimes\ldots\otimes A_m)=\Tr_{\fH_N}(F_N(t)(J_1A_1)\ldots(J_mA_m))
$$
for all $A_1,\ldots,A_m\in\cL(\fH)$.

\begin{Lem}
The Pickl functional satisfies the inequality
$$
\|F_{N:m}(t)-R(t)^{\otimes m}\|_1\le 2\sqrt{2m\Tr_{\fH}(F_{N:1}(t)P(t))}\,,\qquad m=1,\ldots,N\,.
$$
\end{Lem}

\begin{proof}
Call $\cP_-$ the spectral projection on the direct sums of eigenspaces of the trace-class operator $F_{N:m}(t)-R(t)^{\otimes m}$ corresponding to negative eigenvalues. Then the self-adjoint operator
$$
\cP_-F_{N:m}(t)\cP_--\cP_-R(t)^{\otimes m}\cP_-=\cP_-F_{N:m}(t)\cP_--|\cP_-\psi(t,\cdot)^{\otimes m}\ra\la\cP_-\psi(t,\cdot)^{\otimes m}|
$$
must have only negative eigenvalues by definition of $\cP_-$, and is obviously nonnegative on the orthogonal complement of $\cP_-\psi(t,\cdot)^{\otimes m}$ in the range of $\cP_-$. By definition of $\cP_-$, this orthogonal complement 
must be $\{0\}$. Hence $\cP_-$ is a rank-one projection, so that $F_{N:m}(t)-R(t)^{\otimes m}$ has only one negative eigenvalue $\lambda_0$, with all its other eigenvalues $\l_1,\l_2,\ldots$ being nonnegative. Since
$$
\Tr_{\fH_m}(F_{N:m}(t)-R(t)^{\otimes m})=\sum_{j\ge 1}\lambda_j+\lambda_0=0\,,
$$
one has\footnote{This observation is attributed to Seiringer on p. 35 in \cite{RodSchlein}.}
$$
\ba
\|F_{N:m}(t)-R(t)^{\otimes m}\|_1=\sum_{j\ge 1}\lambda_j+|\lambda_0|=2|\lambda_0|=&2\|F_{N:m}(t)-R(t)^{\otimes m}\|
\\
\le& 2\|F_{N:m}(t)-R(t)^{\otimes m}\|_2\,.
\ea
$$
Now $F_{N:m}(t)$ is self-adjoint, and therefore
$$
\ba
\|F_{N:m}(t)-R(t)^{\otimes m}\|_2^2=&\Tr_{\fH_m}((F_{N:m}(t)-R(t)^{\otimes m})^2)
\\
=&\Tr_{\fH_m}(F_{N:m}(t)^2+R(t)^{\otimes m})
\\
&-\Tr_{\fH_m}(F_{N:m}(t)R(t)^{\otimes m}+R(t)^{\otimes m}F_{N:m}(t))
\\
\le& 2-2\Tr_{\fH_m}(R(t)^{\otimes m}F_{N:m}(t)R(t)^{\otimes m})
\\
=&2\Tr_{\fH_m}(F_{N:m}(t)(I_\fH^{\otimes m}-R(t)^{\otimes m}))\,.
\ea
$$
Hence
$$
\|F_{N:m}(t)-R(t)^{\otimes m}\|_1\le 2\sqrt{2\Tr_{\fH_m}(F_{N:m}(t)(I_\fH^{\otimes m}-R(t)^{\otimes m}))}\,.
$$

Since $R(t)=|\psi(t,\cdot)\ra\la\psi(t,\cdot)|$ is a self-adjoint projection
$$
\ba
R(t)\otimes I_\fH^{\otimes (m-1)}-R(t)^{\otimes m}
\\
=
(I_\fH^{\otimes m}-I_\fH\otimes R(t)^{\otimes(m-1)})R(t)\otimes I_\fH^{\otimes(m-1)}(I_\fH^{\otimes m}-I_\fH\otimes R(t)^{\otimes(m-1)})
\\
\le(I_\fH^{\otimes m}-I_\fH\otimes R(t)^{\otimes(m-1)})^2=(I_\fH^{\otimes m}-I_\fH\otimes R(t)^{\otimes(m-1)})
\ea
$$
so that
$$
\ba
\Tr_{\fH}(F_{N:1}(t)R(t))-\Tr_{\fH_m}(F_{N:m}(t)R(t)^{\otimes m})
\\
=\Tr_{\fH_m}(F_{N:m}(t)(R(t)\otimes I_\fH^{\otimes (m-1)}-R(t)^{\otimes m}))
\\
\le\Tr_{\fH_m}(F_{N:m}(t)(I_\fH^{\otimes m}-I_\fH\otimes R(t)^{\otimes(m-1)}))
\\
=1-\Tr_{\fH_{m-1}}(F_{N:m-1}(t)R(t)^{\otimes(m-1)})&\,.
\ea
$$

Since $F_{N}^{in}$ satisfies \eqref{SymDens}, the reduced $m$-particle operator $F_{N:m}(t)\in\cL(\fH_m)$ also satisfies \eqref{SymDens} (with $N$ replaced by $m$), and hence
$$
\ba
\Tr_{\fH_m}(F_{N:m}(t)(I_\fH^{\otimes m}-R(t)^{\otimes m}))\le& 1-\Tr_{\fH}(F_{N:1}(t)R(t))
\\
&+1-\Tr_{\fH_{m-1}}(F_{N:m-1}(t)R(t)^{\otimes(m-1)})
\\
\le&m(1-\Tr_{\fH}(F_{N:1}(t)R(t)))
\\
=&m\Tr_\fH(F_{N:1}(t)P(t))\,,
\ea
$$
by induction, which implies the inequality in the lemma.
\end{proof}

\smallskip
With this lemma, the consequence of the Gronwall inequality above implies that, under the assumptions of Corollary \ref{C-MF},
$$
\|F_{N:m}(t)-R(t)^{\otimes m}\|_1\le\sqrt{8m\a_N(t)}\le 4\sqrt{\frac{m}N}\exp\left(\tfrac{3}\hb\int_0^tL(s)ds\right)\,.
$$
This completes the proof of Corollary \ref{C-MF}.

%\smallskip
%\noindent
%\textbf{Data Availability Statement.}  This manuscript has no associated data.

%\smallskip
%\noindent
%\textbf{Competing Interests.} On behalf of all authors, the corresponding author states that there is no conflict of interest. 

%%%%%%%%%%%%%%%%%%%%%%%%%%%%%%%%%%%%%%%%%%%%%%%%%%%%%%%%%%%%%%%%%%%%%%%%%%%%%%%%%%%%%%%%%%%%%%%%%%%%%%%%%%%%%%%%%%%%%%%%%

%%%%%%%%%%%%%%%%%%%%%%%%%%%%%%%%%%%%%%%%%%%%%%%%%%%%%%%%%%%%%%%%%%%%%%%%%%%%%%%%%%%%%%%%%%%%%%%%%%%%%%%%%%%%%%%%%%%%%%%%%


\begin{thebibliography}{99}

\bibitem{BEGMY}
Bardos, C., Erd\"os, L., Golse, F., Mauser, N., Yau, H.-T.:
Derivation of the Schr\"odinger-Poisson equation from the quantum $N$-body problem.
C. R. Acad. Sci. S\'er. I Math \textbf{334} (2002), 515--520.

\bibitem{BGM}
Bardos, C., Golse, F., Mauser, N.:
The weak coupling limit for the $N$-particle Schr\"odinger equation.
Methods Appl. Anal. \textbf{7} (2000), 275--293

\bibitem{Bove2}
Bove, A., Da Prato, G., Fano, G.:
On the Hartree-Fock time-dependent problem. 
Comm. Math. Phys. \textbf{49} (1976), 25--33

\bibitem{BraunHepp}
Braun, W., Hepp, K.: 
The Vlasov dynamics and its fluctuations in the $1/N$ limit of interacting classical particles. 
Commun. Math. Phys. \textbf{56} (1977), 101--113

\bibitem{Brezis}
Brezis, H.:
Functional Analysis, Sobolev Spaces and Partial Differential Equations.
Springer Science + Business Media, LLC 2011

\bibitem{EY}
Erd\"os, L., Yau, H.-T.:
Derivation of the nonlinear Schr\"odinger equation from a many body Coulomb system. 
Adv. Theor. Math. Phys. \textbf{5} (2001), 1169--1205

\bibitem{FGMouPaul}
Golse, F., Mouhot, C., Paul, T.:
On the mean field and classical limits of quantum mechanics.
Commun. Math. Phys. \textbf{343} (2016), 165--205

\bibitem{FGTPaul}
Golse, F., Paul, T.:
Empirical Measures and Quantum Mechanics: Applications to the Mean-Field Limit.
Commun. Math. Phys. \textbf{369} (2019), 1021--1053

\bibitem{HaurayJabin1}
Hauray, M., Jabin, P.-E.:
$N$-particles approximation of the Vlasov equations with singular potential.
Arch. Ration. Mech. Anal. \textbf{183} (2007), 489--524

\bibitem{HaurayJabin2}
Hauray, M., Jabin, P.-E.:
Particle approximation of Vlasov equations with singular forces: Propagation of chaos. 
Ann. Sci. \'Ec. Norm. Sup\'er. (4) \textbf{48} (2015), 891--940

\bibitem{Hirata}
Hirata, H.:
The Cauchy problem for Hartree type Schr\"odinger equations in weighted Sobolev space.
J. Fac. Sci. Univ. Tokyo Sect. IA, Math. \textbf{38} (1991), 567--588

\bibitem{Kato51}
Kato, T.:
Fundamental Properties of Hamiltonian Operators of Schr\"odinger Type.
Trans. of the Amer. Math. Soc. \textbf{70} (1951), 195--211

\bibitem{Kato}
Kato, T.
Perturbation Theorey for Linear Operators.
Springer Verlag, Berlin, Heidelberg, 1995

\bibitem{KPickl}
Knowles, A., Pickl, P.:
Mean-Field Dynamics: Singular Potentials and Rate of Convergence.
Commun. Math. Phys. \textbf{298} (2010), 101--138

\bibitem{Penz}
Penz, M.:
Regularity for evolution equations with non-autonomous perturbations in Banach spaces.
 J. Math. Phys. \textbf{59} (2018), 103512

\bibitem{Pickl1}
Pickl, P.: 
A simple derivation of mean-field limits for quantum systems. 
Lett. Math. Phys. \textbf{97} (2011), 151--164

\bibitem{RS2}
Reed, M., Simon, B.:
``Methods of Modern Mathematical Physics, Vol. II: Fourier Analysis, Self-Adjointness''. 
Acad. Press, (1975)

\bibitem{RodSchlein}
Rodnianski, I., Schlein, B.: 
Quantum fluctuations and rate of convergence towards mean-field dynamics. 
Commun. Math. Phys. \textbf{291} (2009), 31--61

\bibitem{Serfaty}
Serfaty, S.:
Mean field limit for Coulomb-type flows. (Appendix in collaboration with Duerinckx, M.)
Duke Math. J. \textbf{169} (2020), 2887--2935

\bibitem{Spohn80}
Spohn, H.:
Kinetic equations from Hamiltonian dynamics.
Reviews of Modern Physics, \textbf{63} (1980), 569--615

\end{thebibliography}
\end{document}